\documentclass[10pt, journal,compsoc]{IEEEtran}

\usepackage{amsmath}
\usepackage{algpseudocode}
\usepackage{algorithm}
\usepackage{graphicx}
\usepackage{caption}
\usepackage{booktabs}
\usepackage{amssymb}
\usepackage{amsthm}

\newtheoremstyle{theor}
{6pt}
{6pt}
{\itshape}
{1em}
{\scshape}
{.}
{ }
{}
\theoremstyle{theor}

\newtheorem{lemma}{Lemma}[section]

{\theoremstyle{theor}
	\newtheorem{theorem}{Theorem}[section]
	
}

\ifCLASSOPTIONcompsoc
\usepackage[nocompress]{cite}
\else
\usepackage{cite}
\fi
\IEEEtitleabstractindextext{

	\begin{abstract}
		Topological event detection allows for the distributed computation of homology by focusing on local changes occurring in a network over time.
		In this paper, a model for the monitoring of topological events in dynamically changing regions will be developed.
		Regions are approximated as the connected components of the communication graph of a sensor network, reducing homology computation to graph homology. 
		Betti number differences together with cyclic neighbor-rings are used to categorize topological event types.
		The focus lies on the correct detection of non-incremental (i.e., multiple concurrently occurring) events and the necessary region update process. 
		Network number differences between a network's state before and after events are spread from event nodes into network regions, 
		allowing for the conflict-free updating of regions independent of the update messages' order of arrival.
	\end{abstract}
	
	\begin{IEEEkeywords}
		topological event, homology, whitney triangulation, euler characteristic
	\end{IEEEkeywords}
}

\title{Distributed Monitoring of\\Topological Events via Homology}

\author{
	Vincent T.~H. Knapps and Karl-Heinz Zimmermann
	\\Hamburg University of Technology
	\\21071 Hamburg, Germany
}

\begin{document} 
	
\maketitle

\IEEEdisplaynontitleabstractindextext

%
\IEEEpeerreviewmaketitle


\IEEEraisesectionheading{\section{Introduction} \label{ch1}}

%
%
%
%

\IEEEPARstart{T}{opology} captures a network's essential structure; i.e.,
the connected components and holes of the corresponding communication graph.
These are properties which are of importance not only for the in-network routing of messages but also for the representation of a region monitored by a sensor network. 
Often more relevant than a precise approximation of a monitored region is the development of its structure over time; 
the topological changes a region undergoes.
Environmental monitoring entails the periodic measurement and analysis of
geospatial data. Applications include surveillance with moving sensors, the information exchange in multi-agent systems for the purposes of exploration is an example, 
and the monitoring of dynamic regions; usually the collected data is used for predictions (e.g., for weather forecasts). 
\par
The topic of this paper is the monitoring of topological events in dynamically evolving regions via homology. 
These events represent fundamental changes to the topology of observed regions; 
the formation of holes and the merging of different regions are prominent examples.
Topological event detection is an appropriate method for the distributed computation of a region's topological properties:
A monitoring sensor network can be subdivided into connected components whose sensors share the same readings. 
Together they represent a so-called clique complex for which homology can easily be computed by means of graph homology. 
Moreover, Betti number differences resulting from changes in sensor readings then can be used to detect topological events.
All events are detected locally at event nodes, where the required Betti number differences can be inferred by querying neighboring nodes.
\par
The remainder of this paper is organized as follows:
Section~\ref{ch2} provides general definitions necessary for 
topological event detection, and Section~\ref{ch3} defines the topological event model central to this paper. 
A combination of Betti number differences and a cyclic neighbor-ring is employed for event detection.
As main part of this paper, Section~\ref{ch4} details the distributed region update process initiated upon event detection. 
This update process enables the detection of non-incremental (i.e., multiple concurrently occurring) events. 
Finally, in Section~\ref{ch5}, results of topological event detection for simulated forest fires are presented.

\section{Background}\label{ch2}
This introductory section provides an overview on recent research activities focusing on applying homology for the purpose of network analysis.
Moreover, the necessary background on topology and homology is presented~\cite{hom1,hom2}.

\subsection{Related Work}
There exist several attempts in which efforts were made to employ homology for network analysis. 
In~\cite{hole_detect} a criterion for hole detection in wireless sensor networks (WSNs) with uniform coverage radii is developed:
\par
Two Rips complexes are constructed as bounds for the union of coverage disks. 
Generators of the first homology group which are valid for both are proven to be network holes. 
However, the presented criterion comes with some disadvantages. 
The found generators may be non-minimal, some holes remain undetectable 
(in~\cite{hom_tr_holes} the undetectable holes are identified as triangular holes) 
and all necessary computations are executed centralized.
\par
Another approach~\cite{hom_coverage} is the usage of a coverage criterion.
A network is surrounded by a cycle of boundary nodes. 
When the boundary cycle can be filled in with 2-simplices of network nodes, no holes exist. 
To this end, the second homology group relative to the boundary cycle is computed. 
At least one hole exists if no generator can be found. 
This work was extended in~\cite{hom_distr}, where a distributed algorithm for the coverage criterion is developed.
All above mentioned works focus on the description of one network state. Instead of analyzing one snapshot of a complete network state at a time, the focus
of the research field of \textit{topological events} lies solely on the parts of a network changing over time. 
For example, in~\cite{complex_region} monitored areas are partitioned into different region types. 
Regions of successive network snapshots are matched, and topological events are determined according to observed region type changes. 
Contrary to the previous example, in~\cite{dec_top_event} a distributed algorithm 
for topological event detection is developed. Network nodes approximate the boundary of an areal object and are classified in relation
to their positions at this boundary. The presented algorithm detects events locally at event nodes based on the nodes' and their neighbors'
boundary-state transitions.
\par
All of these works have in common a focus on incremental event detection.
Only one node's reading can change at a time; i.e., all events are detected successively. 
Consequently, the process of region updates, in which region identifiers are distributed in response to topological events, is not considered. 
In view of this observation, the contribution of this paper then is the development of a distributed region update process which allows for the detection of non-incremental events.

\subsection{Homology and Betti Numbers}\label{ch2:betti}
In algebraic topology, the information about a topological space $X$ can be encoded by a chain complex $C(X)$,
which is a sequence of abelian groups $(C_i)_{\geq 0}$ connected by homomorphisms $\partial_n:C_n\rightarrow C_{n-1}$ called boundary operators.
It is necessary that the composition of any two consecutive boundary operators is trivial; that is, $\partial_n\circ\partial_{n+1}=0$ for each $n\geq 0$.
Thus the image of the boundary operator $\partial_{n+1}$ is contained in the kernel of the boundary operator $\partial_n$; 
that is, ${\rm im}(\partial_{n+1})\subseteq{\rm ker}(\partial_n)$.
The elements of $B_n(X) = {\rm im}(\partial_{n+1})$ are called boundaries
and the elements of $Z_n(X)={\rm ker}(\partial_n)$ are called cycles.
\par
The quotient group $H_n(X)=Z_n(X)/B_n(X)$ is called the $n$th homology group of $X$. Elements of $H_n(X)$ are called homology classes and the rank of $H_n(X)$ is called the $n$th Betti number of $X$; it is denoted by $\beta_n(X)$.
The number $\beta_n(X)$ counts the number of $n$-dimensional holes of a topological space $X$.
In particular, $\beta_0$ is the number of connected components of $X$,
$\beta_1$ is the number of one-dimensional holes of $X$, and 
$\beta_2$ is the number of two-dimensional cavities of $X$.
\par
For instance, a graph with $n$ vertices, $m$ edges, and $k$ connected 
components has the 	Betti numbers $\beta_0 = k$, $\beta_1=m-n+k$, and 
$\beta_n=0$ for $n\geq 2$.

\subsection{Euler Characteristic}\label{ch2:euler}
An abstract simplicial complex is a family $\Delta$ of non-empty finite subsets of a set $S$ which is closed under the operation of taking subsets.
The finite sets in $S$ are called faces of the complex.
A complex $\Delta$ is finite if it has finitely many faces.
The dimension of a face $F$ is defined as $|F|-1$.
In particular, the zero-dimensional faces are called the vertices of $S$.
And the dimension of the complex $\Delta$ is given by the largest dimension of any of its faces.
\par
For instance, let $V$ be a finite subset of $S$ of cardinality $n+1$, and let $\Delta=2^V$ be the power set of $V$.
Then $\Delta$ is called a combinatorial $n$-\emph{simplex} with vertex set $V$.
Each combinatorial $n$-simplex has dimension $n$.
In particular, if $V=S=\{0,1,\ldots,n\}$, then $\Delta=2^V$ is called the standard combinatorial $n$-simplex.
\par
The \emph{Euler characteristic} of an abstract simplicial complex $\Delta$ is the alternating sum 
\begin{equation} \label{ch2:eq:euler}
\chi(\Delta) = k_0 - k_1 + k_2 - k_3 + \ldots, 
\end{equation}
where $k_n$ is the number of combinatorial $n$-simplices of $\Delta$.
More generally, the Euler characteristic of a topological space $X$ is the alternating sum
\begin{equation} \label{ch2:eq:euler:betti}
\chi(X) = \beta_0 - \beta_1 + \beta_2 - \beta_3 + \ldots, 
\end{equation}
where $\beta_n$ denotes the $n$th Betti number of $\Delta$.
For abstract simplicial complexes these two definitions will yield the same value for $\chi$.
\par
For instance, a finite connected planar graph $G$ with $n$ vertices, $m$ edges, and $f$ faces including the exterior face has Euler characteristic
$\chi(G) = n-m+f =2$.
In particular, if $G$ is a tree, then $m=n-1$ and $f=1$.
If the graph $G$ has $k$ connected components, then $n-m+f-k=1$.
More generally, the Euler characteristic was first defined for the surface of polyhedra as $\chi=n-m+f$.
In particular, the surface of any convex polyhedron has Euler characteristic $\chi=n-m+f=2$.

\subsection{Clique Complexes}\label{ch2:graph}
Clique complexes form a subclass of abstract simplicial complexes.
The \emph{clique complex} $X(G)$ of an undirected graph $G$ has a combinatorial simplex for each clique of the graph (Fig.~\ref{f-cc}).
Since each subset of a clique is also a clique, the family of sets forms an abstract simplicial complex.
The 1-skeleton of $X(G)$ is an undirected graph whose vertices correspond one-to-one with the 1-element sets in the family
and whose edges are associated one-to-one with the 2-element sets in the family; i.e., the 1-skeleton of $X(G)$ is isomorphic to $G$.
	
\begin{figure}[t]
\begin{center}
\mbox{\includegraphics[scale=0.40]{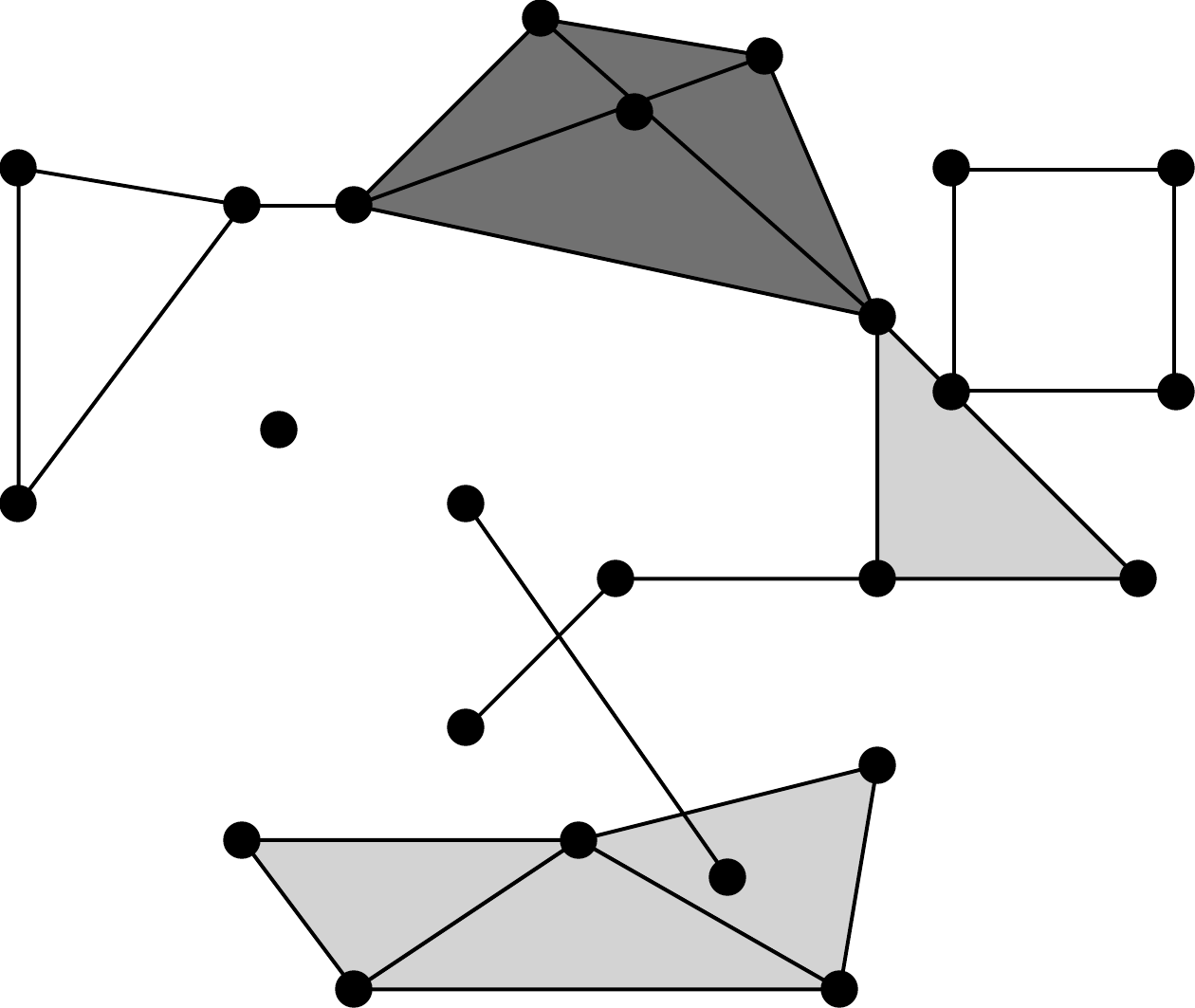}}
\end{center}
\caption{Clique complex of a graph.}\label{f-cc}
\end{figure}

\subsection{Whitney Triangulation}\label{ch2:wt}

A triangulation of a topological space $X$ is an abstract simplicial complex $\Delta$ which is homeomorphic to $X$.
Triangulations are important in algebraic topology since they allow to compute homology and cohomology groups of triangulated spaces.
A \emph{Whitney triangulation} of a compact surface $X$ is an embedding of an undirected graph $G$ onto the manifold such that 
the faces (triangles) of the embedding are exactly the cliques (triangles) of the graph~\cite{whitney}.
The resulting clique complex then is homeomorphic to the surface.
The neighborhood set of a vertex in a Whitney triangulation is either cyclic (in the interior) or forms a path (at the boundary).
A Whitney triangulation is closed exactly when its 1-skeleton is a locally cyclic graph, i.e., each vertex has a cyclic neighborhood structure.
\par
A Whitney triangulation of a compact surface contains combinatorial simplices up to dimension two; i.e., the corresponding clique complex $\Delta_2$
is a 2-clique complex and thus the Euler characteristic for $\Delta_2$ is given by
\begin{eqnarray}
\chi(\Delta_2) = n-m+f=\beta_0-\beta_1,
\end{eqnarray}
where $n$, $m$, and $f$ denote the number of vertices, edges, and faces, respectively.
Thus the number of holes in a Whitney triangulation is given as
\begin{eqnarray}
\beta_1 = -n+m-f+\beta_0.
\end{eqnarray}

\section{Topological Event Detection} \label{ch3} 

This section describes a formal model for topological event detection via homology,
which is based in parts on the research of Farah et al.~\cite{top_event_hom,top_event_ring}.
Differences of Betti numbers \cite{top_event_hom} are used to represent topological changes.
Moreover, network nodes possess a cyclic ordering of their neighbors, the so-called neighbor-ring \cite{top_event_ring}, with which events can be detected as binary patterns.
\par
In the following, the topological space given by a forest fire will serve as an application example for topological event detection. 
To this end, it is assumed that the observed space contains a static sensor network which updates measured values to approximate regions covered by fire. 
The common model makes use of a fire index (FI), a value incorporating information such as humidity, temperature, and wind speed.

\subsection{Sensor Network Model}\label{ch3:sensor_model}

Let $X$ be a bounded region embedded in the Euclidean plane ${\mathbb R}^2$.
Each point $p \in X$ is denoted as a vector $(x_p,y_p,t_p)$, 
where the first two coordinates represent the spatial position and the third one 
represents a measurable scalar value. 
In the context of forest fire monitoring, the value $t_p$ represents the FI-value at the 
given position. Using a threshold value $\theta$, the measured scalar values can be 
discretized into binary values (0 and~1). For this, let $\hat t_p$ denote the binary value 
representing the discretized scalar value at sensor location $p$. The connected components 
with FI-values of~one form to be monitored fire regions.
\par
Let $S$ denote a non-empty finite set of monitoring sensors stationed in $X$.
Consider a Whitney triangulation $\Delta$ of the region $X$ whose vertex set is given by $S$.
The neighborhood structure of each interior node is cyclic, while the neighborhood 
structure of each boundary node forms a path.
By deleting the nodes in the triangulation $\Delta$ (including their edges) which have 
FI-values of~zero, one obtains the clique complex $\Delta_\theta$ of the to be monitored fire 
regions in $X$. We assume that for each sensor $s\in S$ of the network the following data are available:
\begin{itemize}
\item  neighborhood structure $N_s$,
\item  measured FI-value $t_s$ and associated binary value~$\hat t_s$,
\item  component information (number of nodes, edges, and faces) $(n_s,m_s,f_s)$, 
\item  component-ID $c_s$ and sensor-ID $\pi_s$.
\end{itemize}
Note that each sensor in $\Delta_\theta$ possesses information about the connected component to which it belongs. 
This is expressed by the component information and the component-ID. 
Also note that the clique complex $\Delta_\theta$ is embedded in the plane
and so the zeroth Betti number $\beta_0$ is the number of fire regions (connected components) and the first Betti number $\beta_1$ is the number of one-dimensional holes
(Fig.~\ref{f1}).

\begin{figure}[t]
\centering
\includegraphics[width=60mm]{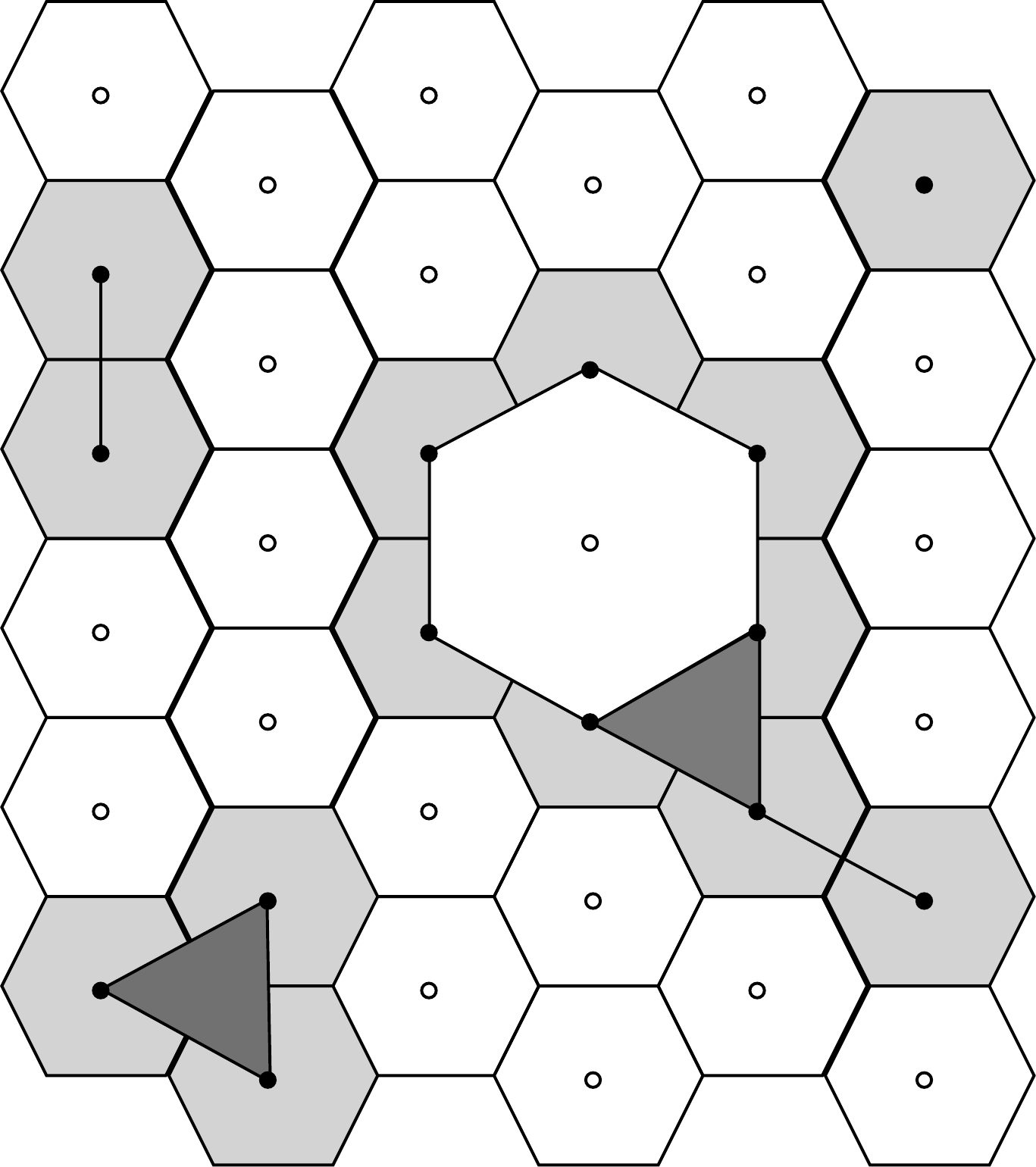}
\caption{ \footnotesize
A sensor network monitoring fire regions indicated by gray areas. 
The sensor $v_0$ lying inside a hole has a FI-value of~0 and cyclic neighbor list $N_0 = (v_1,v_2,v_3,v_4,v_5,v_6)$, the sensors of which all have FI-values of~1. 
The fire region surrounding $v_0$ has eight vertices, nine edges, and one shaded region (triangle). Its Betti numbers are $\beta_0=1$ and $\beta_1=-8+9-1+1=1$.}\label{f1}
\end{figure}

\subsection{Topological Events}\label{ch3:top_event}

A \emph{topological event} is a change of one or more topological invariants (i.e., Betti 
numbers) occurring in a topological space $X$ between two points in time.
Topological events were first defined via Betti number differences in \cite{top_event_hom}.
Write $\beta_n^{(t)}$ for the $n$th Betti number of the network observed at time $t$.
Then the difference between the Betti numbers of the network observed at successive time steps $t$ and $t'$ ($t'>t$) is given by
\begin{equation}
\Delta\beta_k(t,t') = \beta_k^{(t')} - \beta_k^{(t)}, \quad k\geq 0.
\end{equation}
A \emph{positive} topological event can be described as a mapping 
between two topological spaces $\phi: X \rightarrow X'$ which is injective but not surjective (e.g., addition of a new node to $\Delta_\theta$),
whereas a \emph{negative} topological event is a
surjective but not injective mapping (e.g., deletion of a node from $\Delta_\theta$).
\par
In principle, there are nine types of topological events for two-dimensional spaces (Tab.~\ref{t-events}), 
which can be completely described by the first two Betti numbers $\beta_0$ and $\beta_1$.  
The first eight event types are causing a change of topological invariants (Figs.~\ref{ch3:events_start}-\ref{ch3:events_end}),
while the ninth event type of topological invariance describes the case when no topological relevant event occurs, i.e., $\Delta\beta_0 = \Delta\beta_1 = 0$; for a positive event this means a growing region, for a negative event a shrinking region. 
In the following, we will describe how these events can be detected using a neighbor-ring in addition to Betti number differences.

\begin{figure}[t]
	\centering
	\includegraphics[width=30mm]{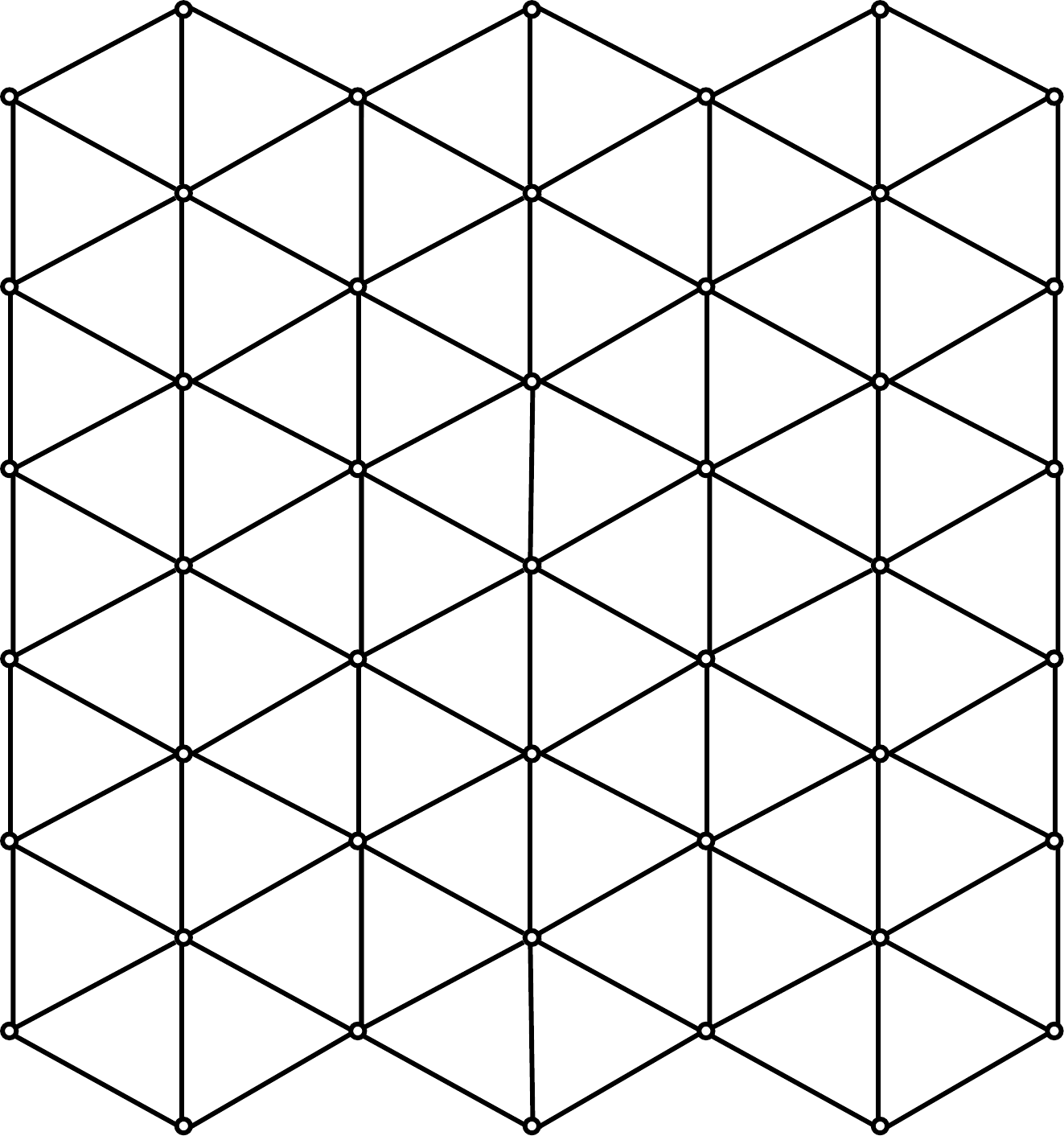}
	\begin{minipage}[b][8em][c]{.11\textwidth}
		\centering
		$\Delta\beta_0 >= 1$\\
		$\Delta\beta_1 = 0$\\
		{\fontsize{45}{45} $\rightleftarrows$}\\
		$\Delta\beta_0 <= -1$\\
		$\Delta\beta_1 = 0$\\
	\end{minipage}
	\includegraphics[width=30mm]{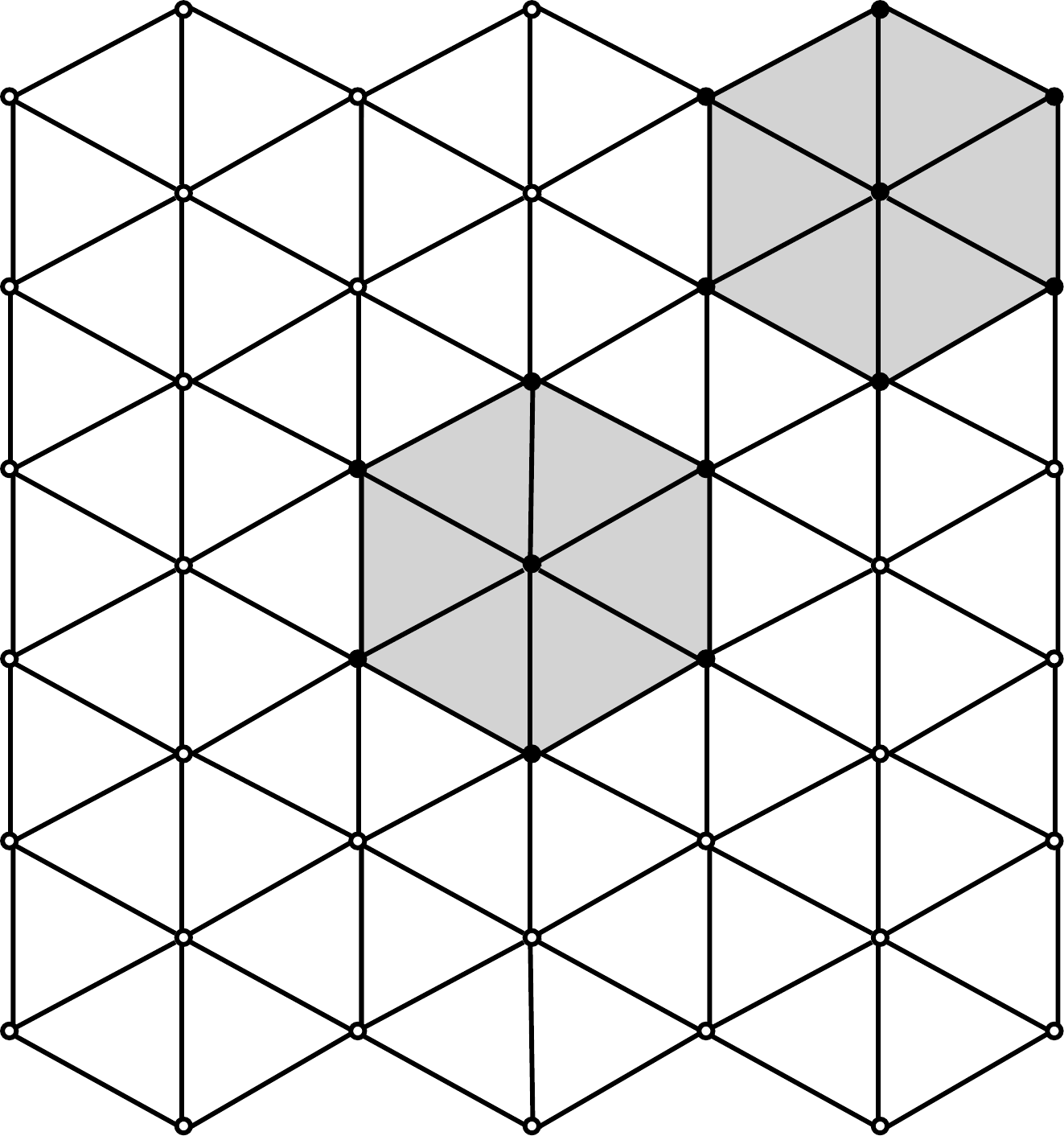}
	\caption{Region-appearance/-disappearance (1,2): Two regions appear/disappear with Betti numbers $\beta_0 = 0$, $\beta_1 = 0$, and $\beta_0 = 2$, $\beta_1 = 0$, respectively.}
	\label{ch3:events_start}
\end{figure}
\begin{figure}[t]
	\centering
	\includegraphics[width=30mm]{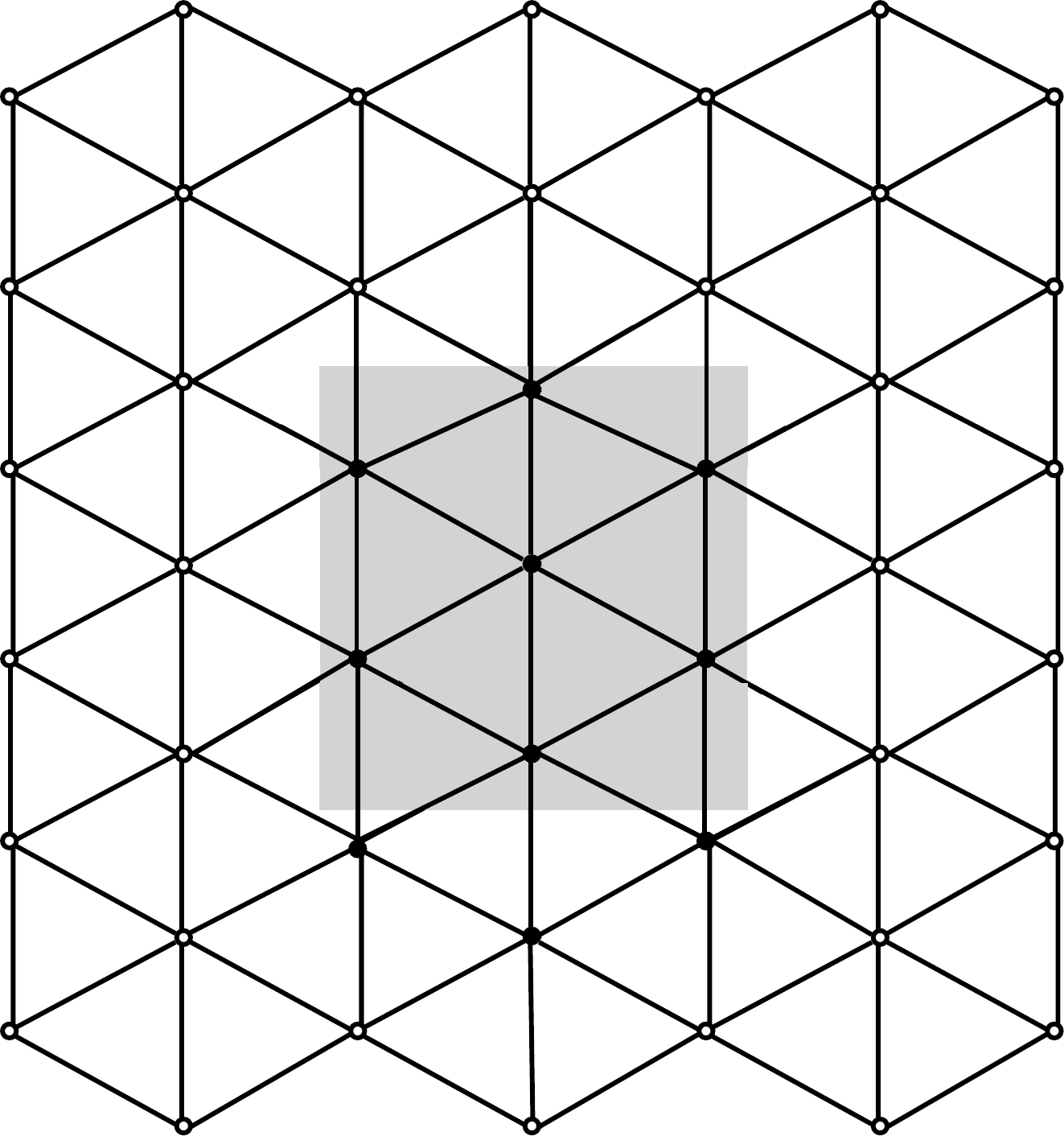}
	\begin{minipage}[b][8em][c]{.11\textwidth}
		\centering
		$\Delta\beta_0 = 0$\\
		$\Delta\beta_1 >= 1$\\
		{\fontsize{45}{45} $\rightleftarrows$}\\
		$\Delta\beta_0 = 0$\\
		$\Delta\beta_1 <= -1$\\
	\end{minipage}	
	\includegraphics[width=30mm]{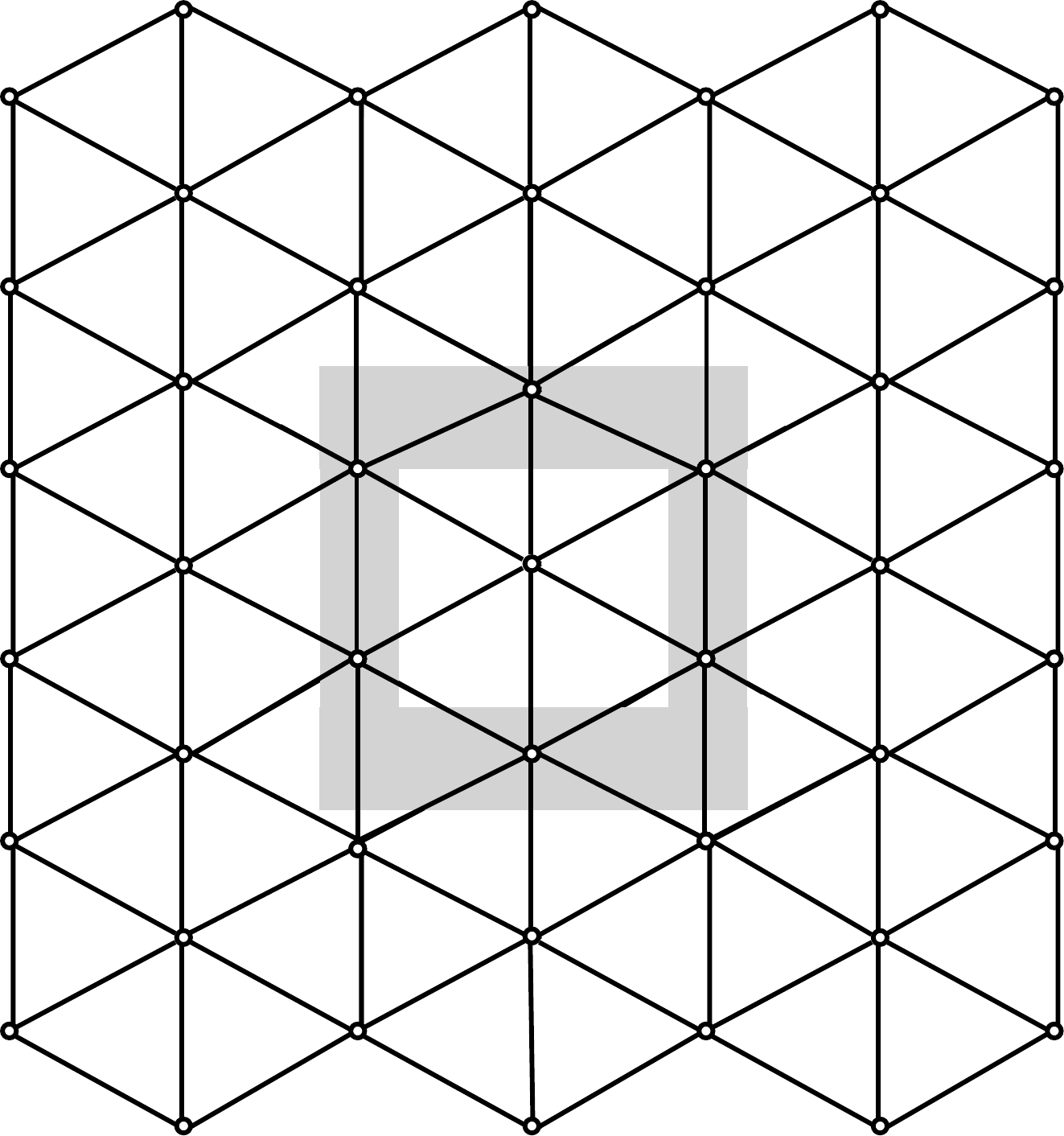}
	\caption{Hole-appearance/-disappearance (3,4): A forming/closing hole with Betti numbers $\beta_0 = 1$, $\beta_1 = 0$, and $\beta_0 = 1$, $\beta_1 = 1$, respectively.}
\end{figure}
\begin{figure}[t]
	\centering
	\includegraphics[width=30mm]{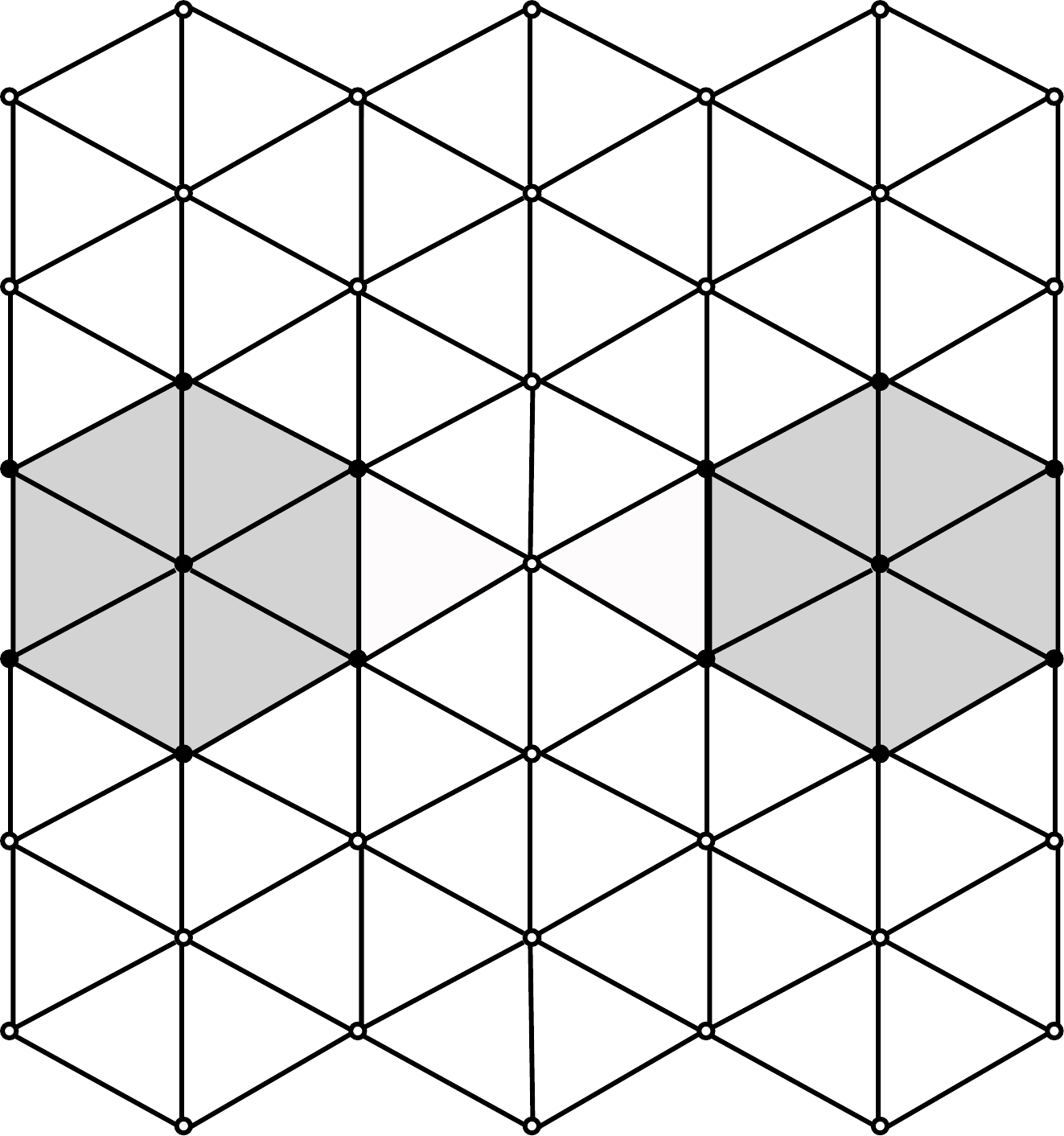}
	\begin{minipage}[b][8em][c]{.11\textwidth}
		\centering
		$\Delta\beta_0 <= -1$\\
		$\Delta\beta_1 = 0$\\
		{\fontsize{45}{45} $\rightleftarrows$}\\
		$\Delta\beta_0 >= 1$\\
		$\Delta\beta_1 = 0$\\
	\end{minipage}	
	\includegraphics[width=30mm]{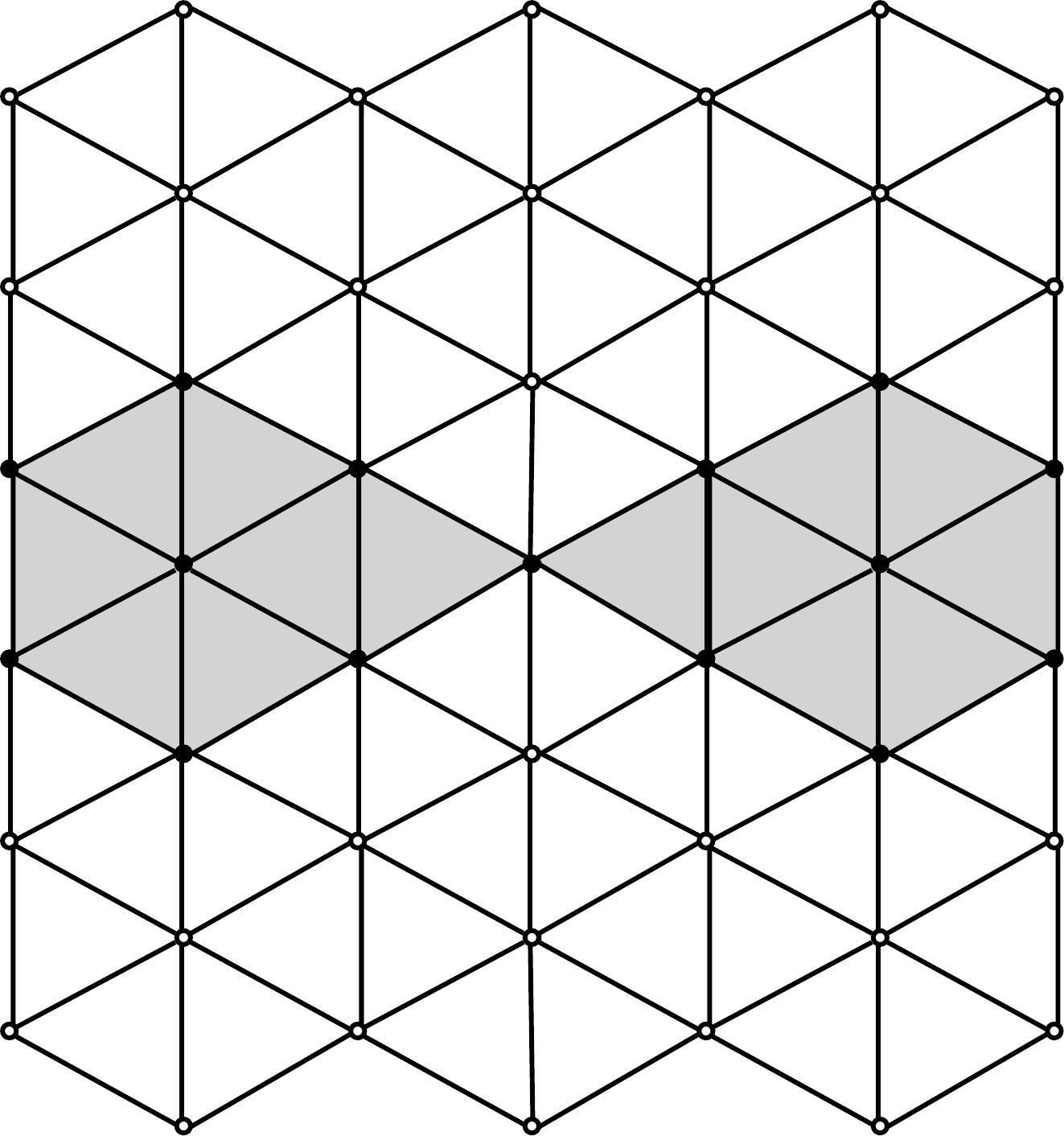}
	\caption{Region merge/split (5,6): Two regions merging/splitting with Betti numbers $\beta_0 = 2$, $\beta_1 = 0$ and $\beta_0 = 1$, $\beta_1 = 0$, respectively.}
\end{figure}
\begin{figure}[t]
	\centering
	\includegraphics[width=30mm]{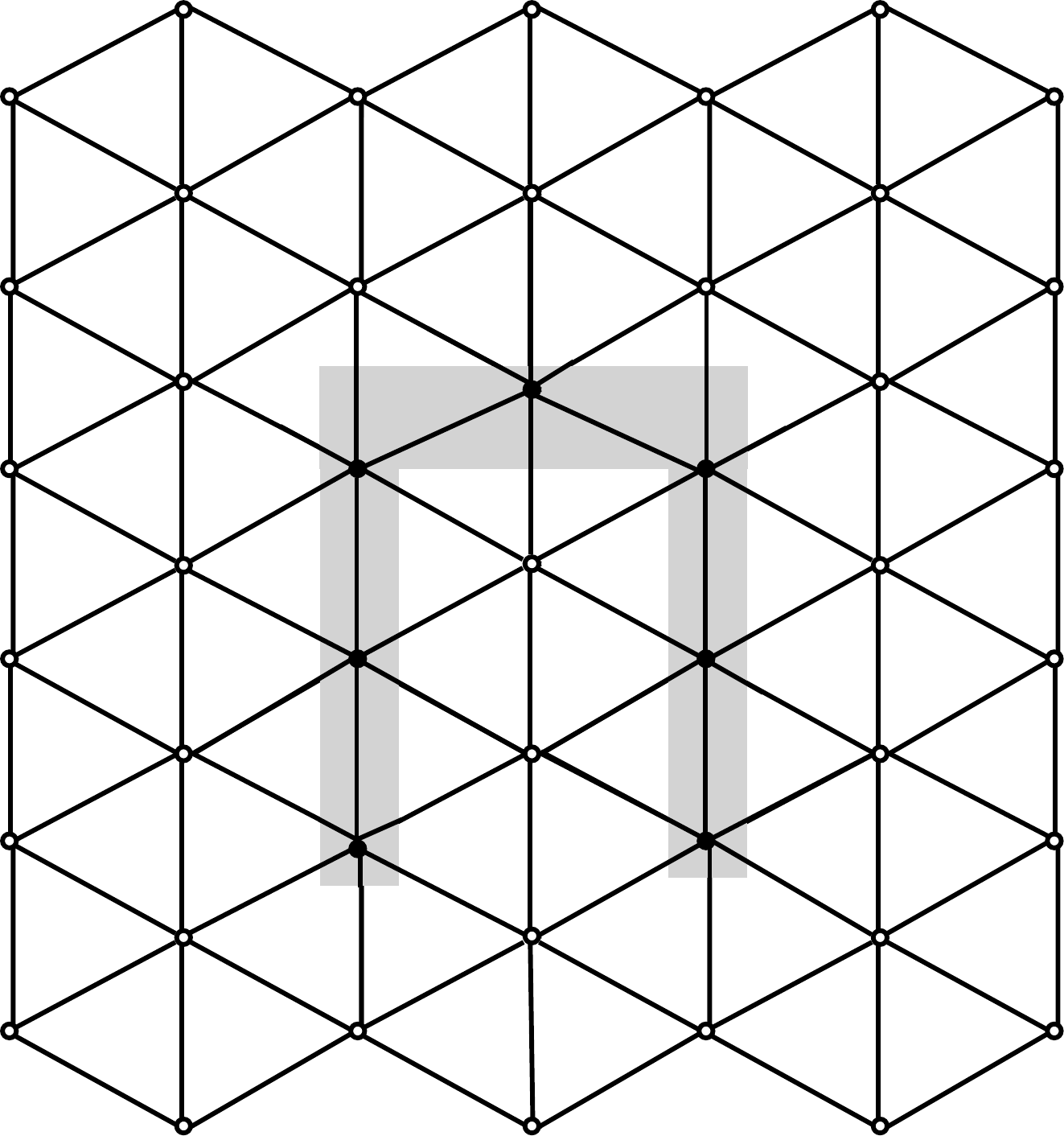}
	\begin{minipage}[b][8em][c]{.11\textwidth}
		\centering
		$\Delta\beta_0 = 0$\\
		$\Delta\beta_1 >= 1$\\
		{\fontsize{45}{45} $\rightleftarrows$}\\
		$\Delta\beta_0 = 0$\\
		$\Delta\beta_1 <= -1$\\
	\end{minipage}	
	\includegraphics[width=30mm]{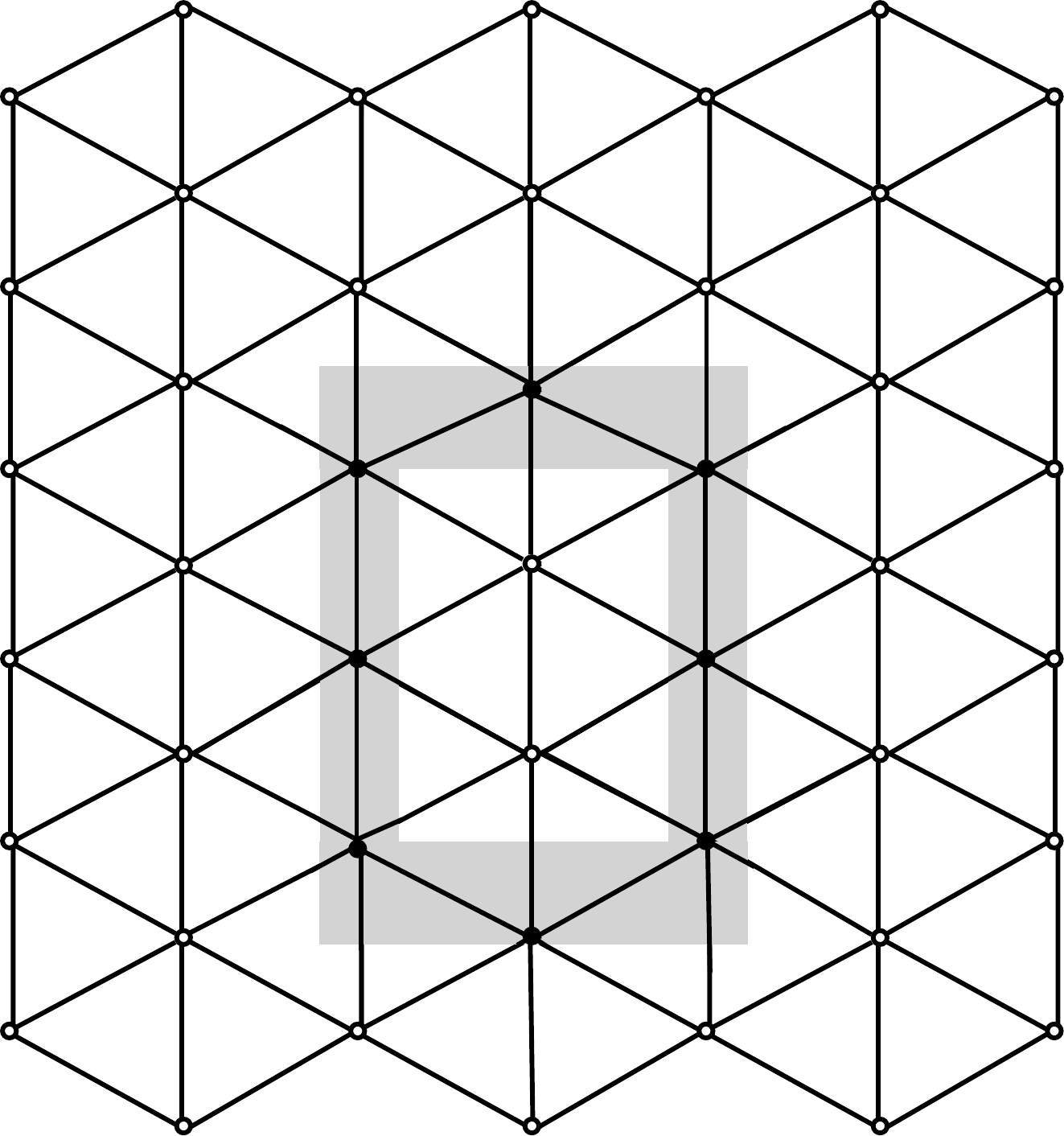}
	\caption{Region self-merge/-split (7,8): A self-merging/-splitting region with Betti numbers $\beta_0 = 1$, $\beta_1 = 0$ and $\beta_0 = 1$, $\beta_1 = 1$, respectively.}\label{ch3:events_end}
\end{figure}

\begin{table}[b]
	\centering
	\begin{tabular}{lllr} 
		\multicolumn{2}{r}{Positive Events} & \multicolumn{2}{l} {Negative Events}\\
		\cmidrule{2-3}                                    
		(1) & Region-Appearance   & Region-Disappearance & (2)\\
		(3) & Hole-Appearance     &  Hole-Disappearance & (4)\\
		(5) & Region-Merge        &   Region-Split & (6)\\   
		(7) & Region-Self-Merge   & Region-Self-Split & (8)\\   
		\hline
		(9) & \multicolumn{3}{c}{Topological Invariance}
		 \\
        \hline
	\end{tabular}
	\caption{Topological event types.}\label{t-events}                                                   
\end{table}

\subsection{Event Detection via Homology}\label{ch3:event_hom}

Basic topological event detection can be achieved by using homology.
To see this, define an {\em event node\/} to be a node in the triangulation $\Delta$ whose FI-value 
changes between two successive time steps.  Furthermore, assume that all event nodes' neighbors are 
non-event nodes.
A {\em positive\/} event occurs if an event node's FI-value $t_s$ 
exceeds $\theta$ ($\hat t_s=1$); otherwise, the event is {\em negative} ($\hat t_s=0$).
In case of a positive event, the corresponding event node is added to $\Delta_\theta$; 
otherwise, it is deleted. Event nodes with associated positive events
are called positive event nodes, and event nodes associated to negative events
are called negative event nodes.
\par
An {\em inner edge\/} is an edge in $\Delta$ linking an event node to one of its neighbors,
and an {\em outer edge\/} is an edge in $\Delta_\theta$ connecting two neighbors of an event node.
Formally, the sets of inner and outer edges of an event node $v\in V(\Delta)$ are defined respectively as follows,
\begin{eqnarray}
{\rm In}(v) &=& \{ \{v,w\}\in
E(\Delta) 
\mid w \in N_v(\Delta_\theta) \},
\end{eqnarray}
and
\begin{eqnarray}
\lefteqn{{\rm Out}(v) =}\\
&&\{ \{u,w\} \in E(\Delta_\theta) \mid u,w \in N_v(\Delta_\theta), u \neq w\}.\nonumber
\end{eqnarray}
The cardinalities of these sets represent the numbers of changed edges $e_{\rm new}=|{\rm In}(v)|$ 
and faces $f_{\rm new}=|{\rm Out}(v)|$ detected by an event node between two successive snapshots 
of the network. 
\par
With $\hat t_s$ capturing the addition/deletion of a vertex (i.e., the event's sign)
as $n_{new} = \hat t_s$, the component information can be updated by adding $n_{new},e_{new},f_{new}$
with appropriate signs to ($n_s,m_s,f_s$). Combined with the number of surrounding components 
$\beta_0$, the updated component data can be used to locally calculate Betti number differences
of the connected component to which the event node belongs.

\subsection{Cyclic Neighborhood Ring}\label{ch3:ring}

Event nodes may have insufficient information to correctly determine a topological event type. 
In case of a positive event the associated event node can 
determine its number of surrounding components $\beta_0$ using the component-IDs of 
neighboring nodes. However, in case of a negative event, this component-ID query will 
fail since all neighbors previously were part of the same component as the event node; 
i.e., all nodes will necessarily have the same component-IDs.
In order to compute the zeroth Betti number for negative events, we introduce an  
additional data structure which was used in \cite{top_event_ring} for event detection, 
the so-called neighbor-ring.
\par
The Whitney triangulation $\Delta$ defining a network's 
communication graph determines a cyclic ordering for each interior node's neighbors.
FI-values of all direct event node neighbors are collected in a list
called {\em neighbor-ring\/} 
and sorted by the event node's cyclic ordering. A continuous block of ones in the 
neighbor-ring represents a {\em ring component}; 
boundary nodes of $\Delta$, having non-cyclic neighbor-rings, always 
separate entries of ones
at the start/end of the list as two different ring components.
The number of ones in the neighbor-ring is equal to $e_{new}$, while 
the number of ring components $r_c$ allows the computation of $f_{new}= e_{new}-r_c$.
\par
Generally, the number of ring components can be assumed to be equal to 
the number of different connected components surrounding an event 
node, and can be used as a replacement for the zeroth Betti number when detecting negative events.
But this assumption only holds true for split events. Self-split events, seen from an event 
node's perspective, involve multiple ring components, yet only one region actually 
exists. Split and self-split events are therefore indistinguishable 
(Fig.~\ref{ch3:fig:ring}). Section~\ref{ch4:self_split} will provide a 
(partial) solution for the detection of self-split events.

\begin{figure}[b]
	\centering
	\includegraphics[width=35mm]{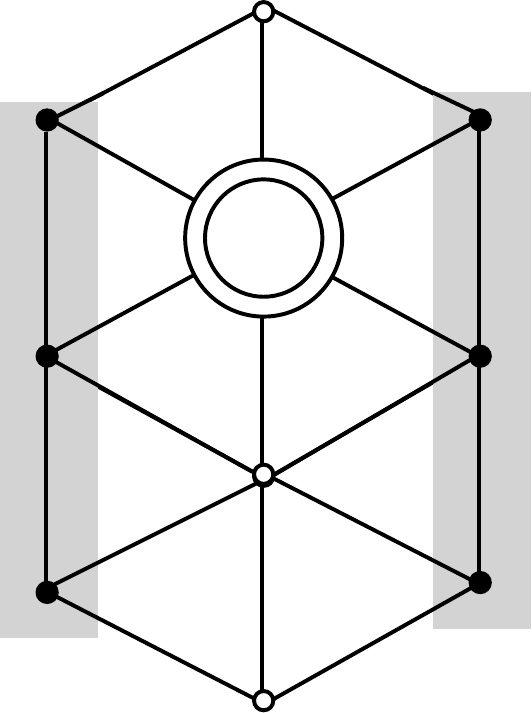}
	\hspace{5mm}
	\includegraphics[width=25mm]{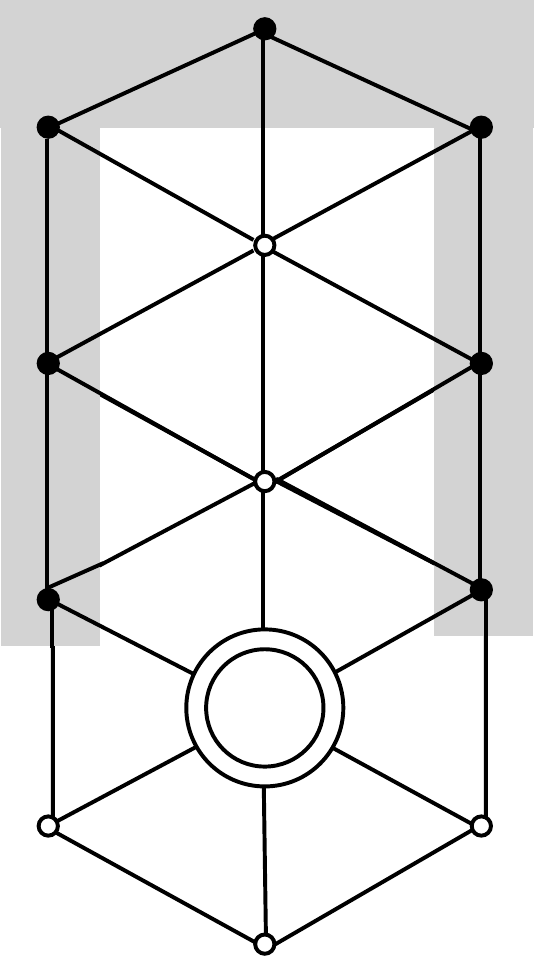}
	\caption{
		\footnotesize
		The triangulation to the left demonstrates the result of a split event.
		The marked event node has neighbor-ring $[1,1,0,1,1,0]$.
		The other triangulation demonstrates a self-split event; the event
		node's neighbor-ring is $[1,0,1,0,0,0]$.}
	\label{ch3:fig:ring}
\end{figure}

\subsection{Ring Query} \label{ch3:ring_query}
Event nodes query their neighbors to both update their neighbor-rings for event detection 
and attain the component data
of surrounding ring components necessary for region updates (Sect.~\ref{ch4:r_update}). 
As neighboring nodes belonging to the same ring component share their component
data, only one representative node per ring component has 
to be queried by an event node. For this, a {\em ring query\/} can be executed:
\par
At first the event node creates an \emph{event token} with its sensor-ID as content. This token is
passed to the neighbor associated with the first entry in the neighbor-ring.
When the queried neighbor has $\hat t_s=0$, the ring query is rejected and the next node in the
neighbor-ring is queried. Otherwise, the representative node reports its component information to 
the event node before starting two {\em query chains} - one for each neighbor
shared with the event node - in which the event token is passed around the ring of event node neighbors 
simultaneously in two directions. 
Event tokens are passed as long as the receiving sensors have FI-values 
of one. These tokens serve to detect when the cyclic ring of neighbors is completed; both query 
chains will eventually reach a node already possessing an event token. When a token cannot be passed
any further, i.e., the next node rejects the token ($\hat t_s=0$) or the cycle is completed,  
the corresponding nodes report back the so-called \emph{chain ends} to the event node, ending 
the query of one ring component.
A chain end consists of the sensor-IDs of the last two nodes in a query chain. 
The two received chain ends are used by the event node 
to determine which entries in the neighbor-ring must be set to one for the previously queried ring component. 
These are precisely the neighbor-ring entries inbetween the entries associated to the two chain ends 
by the event node's cyclic order.
\par
This process is repeated for the next ring component, whose first queried sensor is located at least 
two positions behind the last queried ring component in the neighbor-ring. The ring query ends when all
neighbors were queried directly by the event node, or indirectly through a ring component
query (Fig.~\ref{f-ring-query}).
\begin{figure}[t]
	\centering
	\includegraphics[width=60mm]{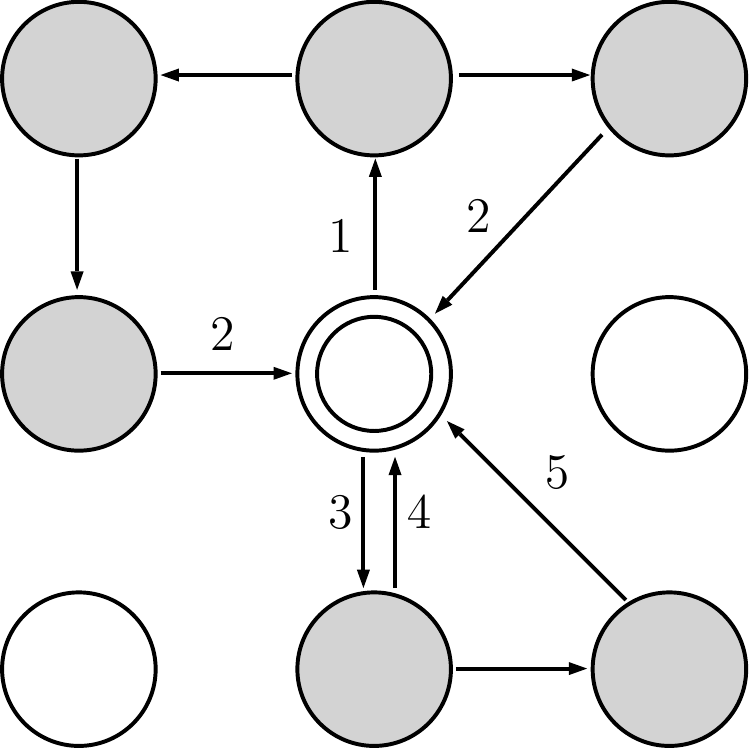}
	\caption{A ring query with two ring components. Numbers indicate the temporal 
	order of events; 
	neighbor-ring: $[1_1,1_2,1_3,0_4,1_5,1_6,0_7,1_8]$,\\
	chain ends: \{(8,1),(3,2)\},\{(-,5),(6,5)\}}\label{f-ring-query}
\end{figure}

\subsection{Event Decision Diagram}\label{ch3:event_tree}

The event decision diagram (Fig.~\ref{ch3:fig:dec}) illustrates the decision process of an event node 
for topological event detection. Instead of directly computing Betti number differences, an event node's neighbor-ring is used to infer
the differences and distinguish event types. This approach allows event detection independent 
of the component numbers. Additionally, the actual Betti number differences can be computed and compared
with the inferred values; conflicting values indicate erroneous component data.
\par
Region-Appearance/-Disappearance (1,2) and Hole-Appearance/-Disappearance (3,4) events can be inferred directly from the cyclic neighbor-ring. 
In these cases, the whole neighbor-ring is either a sequence of ones or a sequence of zeros - boundary nodes, having non-cyclic neighbor-rings, cannot detect the events (3,4). 
A self-merge (7) event can be identified when the event node detects multiple ring components but only receives one component-ID from all of its neighbors.
Similarly, merge events (5) are identified when not only multiple ring components but also multiple different component-IDs are received. In particular, a positive $\Delta\beta_1$ value at the merge node implies that a combined merge/self-merge event has happened.
And split events (6), which are indistinguishable (Sect. \ref{ch3:ring}) from self-split events (8), are
identified when the event node is surrounded by multiple ring components.
Each detectable event type in the diagram is annotated with the corresponding Betti number differences. 
All with inequalities listed event types require actual Betti number computations to
determine the precise numbers of appearing/vanishing regions/holes.
\begin{figure}[t]
	\centering
	\includegraphics[width=80mm]{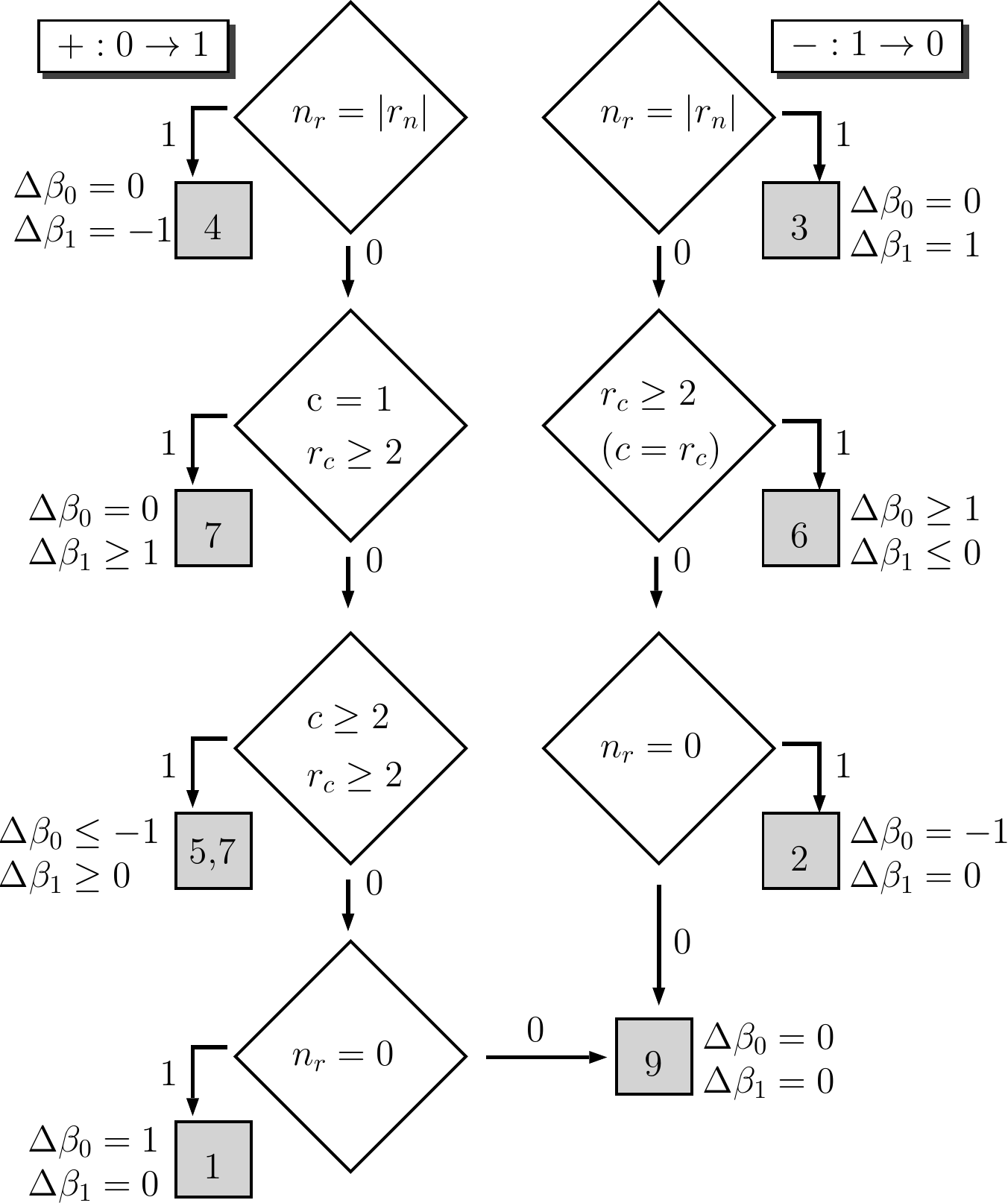}
	\caption{
		\footnotesize
		A diagram illustrating an event node's decision process
		for event detection
($r_n$: neighbor-ring, $n_r$: number of ones in $r_n$, $c$: number of component-IDs, $r_c$: number of ring components).
		Positive events are listed to the left and negative events 
		are found to the right. The numbers in the gray-marked event boxes represent
		event types corresponding to the ones listed in Tab.~\ref{t-events}.
	}\label{ch3:fig:dec}
\end{figure}

\section{Distributed Event Monitoring}\label{ch4}

The previously introduced event detection model (Sect.~\ref{ch3:event_hom})
implicitly relied on the assumption of all network nodes but one event node having stable 
FI readings. By itself this approach is only applicable for incremental event detection. 
The following section will extend this model with a \textit{region update} process 
allowing for the distributed monitoring of non-incremental events.

\subsection{Region Updates}\label{ch4:r_update}

During the process of monitoring an areal object for topological changes, regular updates of 
the monitoring network's sensor data are necessary. 
Let $S_{t}$ denote the network's state detected at time $t$:
It is assumed that all sensors periodically update their FI-values at the start of each sample
time interval $(t,t')$, where $t'>t$. For this, all sensors are assumed to have globally synchronized
clocks. 
\par
Sensors which detect a change in their FI-values, i.e., event nodes, determine the topological
event types and initiate \textit{region updates}.
In response to detected topological events, each node residing inside a connected component of $\Delta_\theta$
must update its component information $[n_s,m_s,f_s]$, as well as its component-ID $c_s$.
To this end component number differences between $S_t$ and the previous state $S_{t-1}$
are distributed by event nodes into their surrounding components, changing the network's state from $S_{t-1}$ 
to $S_{t}$.
\par
During rinq queries event nodes collect component data $[n_c,m_c,f_c]$ from each of their surrounding 
components $c$ (Sect.~\ref{ch3:ring_query}). 
Depending on the detected event type, either the event node's update list $[n_{new},e_{new},f_{new}]$
(Sect.~\ref{ch3:event_hom}), or the queried components' data are used to create lists of component number 
differences for each surrounding component. 

\subsubsection{Region Update Messages}
Region updates are achieved by spreading \textit{region update messages} from event nodes
into regions of $\Delta_\theta$. 
We consider a simplified communication model:
Update messages are spread unidirectionally by each component node starting at the event nodes.
An \textit{update node} is a network node which is receiving update messages.
Update nodes pass received event messages to all of their neighbors.
Zero nodes ($\hat t_s=0$) neither process nor pass received update messages.
Furthermore, we assume that the communication is without errors.
$$\begin{array}{ll}
\mbox{\tt event-ID}: & [\pi_e,i_e]\\
\mbox{\tt update}:   & [n_e,m_e,f_e,k_e,c_e]
\end{array}$$
Region update messages consist of unique event-IDs as message headers, 
and lists of component number differences as update message contents. 
An {\em event-ID\/} contains two parts:
\par
The sensor-ID $\pi_e$ of the event node and the {\em event number\/} $i_e$; each sensor counts 
its number of detected events and assigns each created update message an appropriate event number. 
This definition of an event-ID guarantees that no network node applies the same update more than once; 
update nodes store received event-IDs and reject known updates: $[\pi_{e_1},i_{e_1}] \neq [\pi_{e_2},i_{e_2}]$.
Event nodes store their created event-IDs before sending update messages, preventing them from
processing their own update messages. An update message's contents consist of a list of network numbers to update.
These are the component number differences which each update node adds to its
component information to apply the update. Component numbers are followed by a 
\textit{target component-ID} $k_e$ together with a new component-ID $c_e$
if the corresponding event concerns more than one network component (e.g., merge events). 
New component-IDs are always processed, even if the event-ID
is known (Sect.~\ref{ch4:self_split}).
\par
Component-IDs consist of sensor-IDs and event-IDs of the creating event nodes, i.e, 
$c_e=[\pi_s,(\pi_e,i_e)]$ (from here on only $\pi_s$ is used when referring to component-IDs).
Sensors store component-ID values for both the current and the previous sample interval. 
The latter values are stable during the current sample interval and can be compared against the 
target-ID values of update messages; they indicate intended regions for region update messages. Region 
updates then can be selectively applied only to nodes which were part of the component $k_e$ during the 
previous network state $S_{t-1}$.
\par
There exist three different update message types: Normal updates messages are sent following topological events
which involve exactly one region. Merge update messages are sent into each component connected via
the event node after a merge event. And split update messages are sent into event nodes' surrounding
components after the detection of split events.
In the following the creation of these update message types by event nodes and the subsequent region
update process will be described in detail (Sect.~\ref{ch4:normal_update}--\ref{ch4:split_update}).

\begin{algorithm}[t]
	\begin{algorithmic}[1]
		\footnotesize
		\Function{NormalUpdate}{int $r$}
		\If{${\rm FI}=1$}
		\State $m_e$ $\gets$ $|{\rm In}(e)|$;
		\State $f_e$ $\gets$ $|{\rm Out}(e)|$;
		\Else
		\State $e_e$  $\gets$ $-|{\rm In}(e)|$;
		\State $f_e$ $\gets$ $-|{\rm Out}(e)|$;		
		\EndIf
		\State \Call {spreadUpdate}{$r$, $i_s$, [$({\rm FI}\;? \; 1 \; : \; -1)$, $m_e$,$f_e$]};
		\State $i_e \gets i_e + 1$;
		\EndFunction
	\end{algorithmic}
	\caption{Normal updates are spread through one representative node into a component.}\label{ch4:alg:normal_update}
\end{algorithm}

\subsubsection{Normal Update} \label{ch4:normal_update}
Updates resulting from hole-appearance/-disappearance, region-appearance/-disappearance, and topological 
invariant events fall under the category of normal updates. 
In all aforementioned cases the event node is surrounded by at most one fire region. 
Therefore, it is sufficient to transmit a region update message to one of the event node's
neighbors $r$ lying inside $\Delta_\theta$ (Alg.~\ref{ch4:alg:normal_update}).
\par
The component number difference list for a \textit{normal update message} is created by
multiplying the event node's update list $[n_{new},e_{new},f_{new}]$ with the event's sign. The
normal region update then is applied by adding these values to the component information 
$[n_s,m_s,f_s]$ at each node processing the update message. In case that no surrounding region exists, 
the event node updates its component-ID to its sensor-ID (region-appearance), or skips the region update
process altogether (region-disappearance). Otherwise, the component-ID of the surrounding region is
assumed by the event node; no new component-ID has to be added to the region update message. 

\begin{lemma}\label{ch4:lemma:normal_update}
Normal updates correctly update regions surrounding event nodes, independent on whether other events
occur concurrently in these regions. 
\end{lemma}
\begin{proof}
As update messages are passed to each neighbor inside the component, every node
of the component surrounding the event node will receive the normal update message.
Concurrently merged regions receive normal updates via merge event nodes.   
\par
Update messages from other event nodes cause no data conflicts. Independent on the 
processing order at each component node, the overall sum of component number differences
will reflect all occurred events. When each node has processed 
the normal update, the region is updated in response to
the associated event.
\end{proof}

\begin{algorithm}[t]
	\begin{algorithmic}[1]
		\footnotesize
		\Function{MergeUpdate}{list$<$int$>$ compList, list$<$int$>$ ringNodes}
		\State $\mathtt{add} \; compList \; \mathtt{to:} \; merge\_tokens$;\label{ch4:alg:merge_tokens}
		\State $eN \gets i_e$;
		\State $i_s \gets i_e + 1$;
		\For {$r \in {\rm ringNodes}$} 
		\For {$(c \in {\rm compList} \; | \; c \neq c_r)$}\label{ch4:alg:comp}
		\State\Call{spreadUpdate}{r,$i_e$,[$n_c,m_c,f_c,c_r,\pi_e$]};
		\State $i_s \gets i_s + 1$;
		\EndFor
		\State \Call {spreadUpdate}{r,eN,[1,$|{\rm In}(e)|$,$|{\rm Out}(e)|$]};
		\EndFor
		\EndFunction
	\end{algorithmic}
	\caption{Merge update messages are spread into each component; 
		one message per component $c$ merged into the ring component's region $r_c$ plus 
		an additional normal update message for the event node's numbers.} \label{ch4:alg:merge_update}
\end{algorithm}

\subsubsection{Merge Update} \label{ch4:merge_update}
After the detection of a merge event, the event node sends \textit{merge update messages}
into each surrounding component. For each component $c$ merged into the new merged region,
the component numbers of all other merged components must be added to all component nodes.
Therefore the component information $[n_c,m_c,f_c]$ is used as component number difference list
for merge update messages. Additionally, the data ($\pm1$, $e_{new}$,$f_{new}$) directly added
by a merge event node to the merged region must be added to all surrounding components' numbers. 
Each merge event node transmits one additional normal update message with this data into each component
(Alg. \ref{ch4:alg:merge_update}).
\par
Merge update messages contain the IDs $k_e,c_e$ in addition to the 
component number differences contained in normal updates messages. 
A merge update message is processed only if it originated from a different region 
($c_u \neq k_e$) than the processing update node~$u$. 
In addition, a sensor maintains a list of merged component-IDs, the \textit{merge tokens} 
list, to distinguish which components were already merged during the current sample interval.
Only merge updates with target IDs not already contained in the merge tokens list
are processed. Event nodes detecting merge events fill their merge tokens lists with 
IDs of all their surrounding components (line~\ref{ch4:alg:merge_tokens} of Alg.~\ref{ch4:alg:merge_update}). 
\par
Having different component-IDs before the merge event, a new component-ID $c_e$ 
for the merged regions after the event is required. The event node's sensor ID is used for that purpose:
For each component surrounding the event node, the component-ID $c_e$ of the first split or merge 
update message to reach an update node is used as the new component-ID. Otherwise only a lower 
value is accepted as a new component-ID. Components after merge/split events will assume the lowest
ID of all update messages sent by participating merge/split event nodes in the region. 
\par
Although the components surrounding an event node can be differentiated during merge 
events, merge update messages are spread into each ring component via representative nodes
$r$ to guarantee that each node in the region receives all updates.
For instance, when a combined merge/self-merge event occurs, a self-split 
event could happen concurrently in the same region. Would only one representative node of a component 
be informed of this event, at least one part of the new merged region would not receive the 
merge update message.

\begin{figure}[h]
	\centering
	\includegraphics[width=75mm]{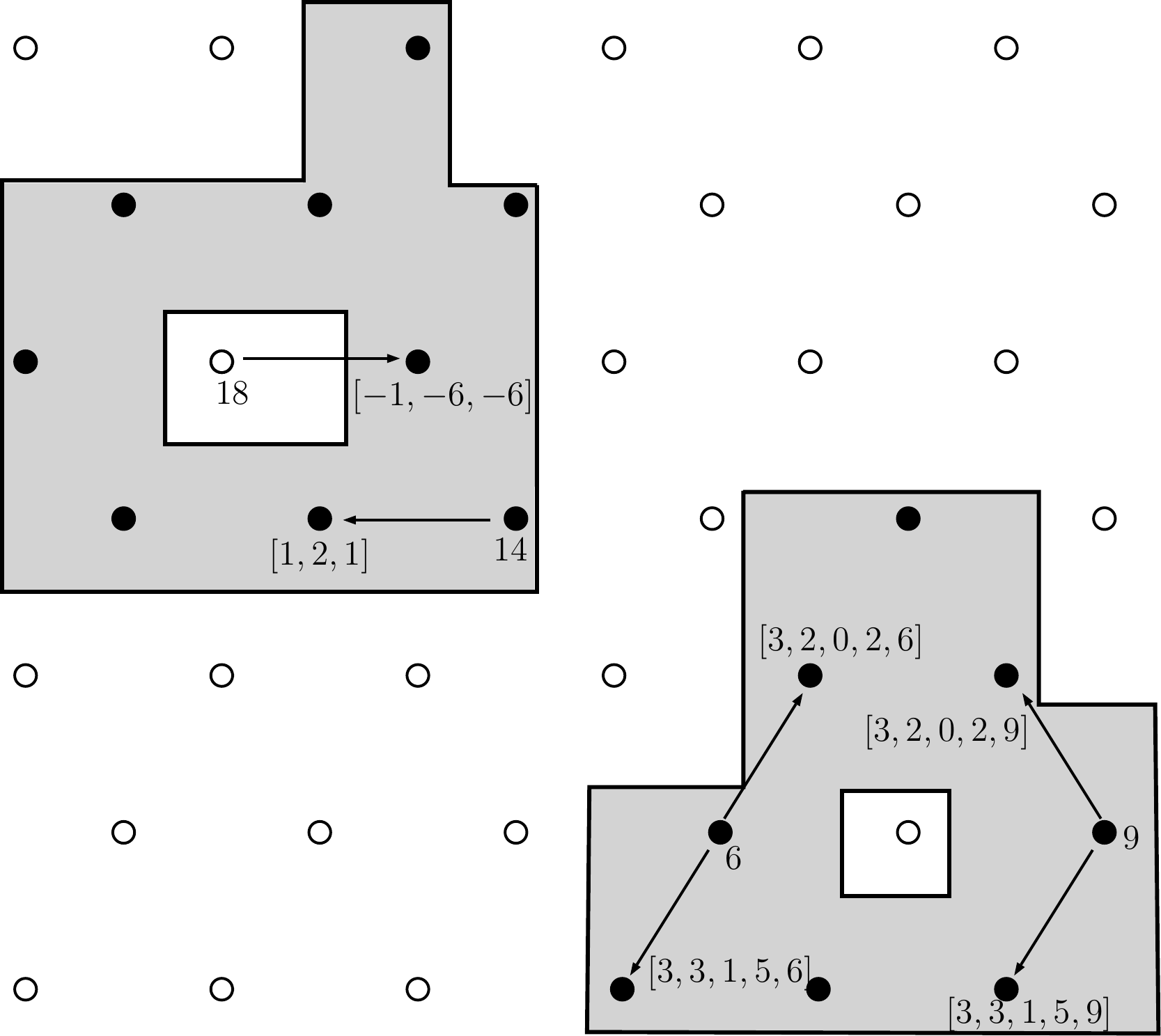}
	\caption{A sensor network with four concurrently detected events. 
		The event nodes 14 and 18 spread \textit{normal update} messages, while the nodes 6 and 9 spread \textit{merge update} messages into the regions 2 and 5.}\label{ch4:fig:normal}
\end{figure}

\begin{lemma}\label{ch4:lemma:merge_update}
Merge updates correctly update regions after merge events,
independent on whether other events occur concurrently inside the merged region. 
\end{lemma}
\begin{proof}
The sums $[\sum_c n_c,\sum_c m_c,\sum_c f_c]$ of component numbers of each merged
component $c$, with the addition of the merge event node's update numbers 
($\pm1$, $e_{new}$,$f_{new}$), represent the component numbers for a merged region.
Merge updates apply these numbers at each node of a merged region. Regions concurrently merged
at other merge event nodes either add to the overall component sums, or represent regions
which are already part of the component sums. Target component-IDs $k_e$ allow for
the distinction between different merge events of the same region:
\par
Two regions can merge at different points at the same time. Without target-IDs the merge event nodes 
would send the same merge update with different event-IDs, and the same merge update would be 
added more than once at each component node. To avoid this possible conflict target-IDs are sent by 
event nodes and saved by each node receiving update messages. Only one merge update message per 
target-ID is applied in one sample interval at each update node.
\end{proof} 

\subsubsection{Split Update}\label{ch4:split_update}
Updates resulting from split, self-split and self-merge events all fall under the category of 
split updates. Self-merge updates, although categorized under split updates, are applied by
sending normal update messages. Contrary to normal updates, self-merge updates must be 
transmitted into each ring component of an event node and are therefore handled 
together with split messages.
\par
After a split event no information on the number of
lost nodes, edges or faces of the split regions is available. 
Instead of transmitting differences to update the component numbers, 
a complete recomputation of the split components is necessary. 
\textit{Split update messages} are sent to one representative node per ring 
component (Alg.~\ref{ch4:alg:split_update}). The reasoning is the same as for merge updates; concurrent
events can cause ring components to become disconnected.
Split updates messages, like merge update messages, contain two component-IDs. 
One for the new component-ID $c_e$ after the split, and one to indicate the targeted region $k_e$. 
The target component-ID serves to indicate which nodes are part of the split
regions. Split updates are executed in two phases: 
\par
In the first phase, the split event node transmits split update messages which contain 
the component numbers of the component previously surrounding the event node with negative signs. 
Each node $u$ receiving a split update message adds these numbers to its component numbers.
Additionally, the split updates' target-IDs are added to the \textit{split tokens} lists of each node.
During one sample interval only one split update message per target-ID is processed, 
guaranteeing that each split component's numbers are subtracted only once when multiple split
event nodes exist in the same region. 
\par
In the second phase, split update nodes $u$ send \textit{split update event 
messages} to recompute the split components. Split update event messages are only 
created by nodes which lie inside the region previously surrounding the event node; i.e., $c_u = k_e$. 
A split update event message is nothing but a normal update message with the following contents:
$$\begin{array}{ll}
\mbox{\tt event-ID}: & [\pi_e,i_e]\\
\mbox{\tt update}:   & [1,e_{\rm new}/2,f_{\rm new}/3]
\end{array}$$
For the computation of $e_{\rm new},f_{\rm new}$ additional ring queries at
each split update node are necessary.
\par
Two special cases must be accounted for when recomputing split regions: 
Neighboring positive event nodes must not be counted for
inner and outer edge numbers, they already sent
these numbers via own event update messages into
the component. Neighboring negative normal event nodes (not split event nodes)
send their inner/outer edge numbers as negative component difference numbers 
into the component. Therefore each split update event
node increments its update numbers with neighboring
negative event nodes' inner/outer edge numbers, and decrements
with positive event nodes' inner/outer edge numbers. Algorithm 4
describes the split update event process in detail.

\begin{algorithm}[b]
	\begin{algorithmic}[1]
		\footnotesize
		\Function{SplitUpdate}{list$<$int$>$ ringNodes, bool isSelfMerge}
		\For {$r \in {\rm ringNodes}$}
		\If {$\neg {\rm isSelfMerge}$}
		\State 
		\Call{spreadUpdate}{r,$i_e$,[-$n_e$,-$m_e$,
			-$f_e$,$c_e$,r]};\label{ch4:alg:split_update:ID}
		\Else
		\State \Call
		{spreadUpdate}{r,$i_e$,[1,$|{\rm In}(e)|$,$|{\rm Out}(e)|$]}		
		\EndIf
		\EndFor
		\State $i_e\gets i_e + 1$;
		\EndFunction
	\end{algorithmic}
	\caption{Split updates are spread into each ring component and 
		initiate a recomputation of the components surrounding the event node.} 
	\label{ch4:alg:split_update}
\end{algorithm} 

\begin{algorithm}[t]
	\begin{algorithmic}[1]
		\Function{SplitEvent}{list$<$int$>$ cyclic\_order}
		\State \Call {RingQuery}{cyclic\_order};
		\State $n_{\rm new}$ $\gets$ 1;
		\State $e_{\rm new}$ $\gets$ $|{\rm In}(e)|$ $*\;1/2$;
		\State $f_{\rm new}$ $\gets$ $|{\rm Out}(e)|$ $*\;1/3$;
		\For {$n \in {\rm neighbor\_ring}$}
		\If {${\rm pos}(n)$} 
		\State $e_{\rm new}$ $\gets$ $e_{\rm new}-1/2$;
		\State $f_{\rm new}$ $\gets$ prevOrNextFI? $f_{\rm new}$ 
		\hspace*{46pt}$-((1/2 * 1/3)\,* f_{\rm new})$ : $f_{\rm new}$;
		\ElsIf {${\rm neg}(n)$}
		\State $n_{\rm new} \gets n_s$ + (1/n.innerEdges);
		\State $e_{\rm new} \gets e_{\rm new} + 1$;
		\State $f_{\rm new}$ $\gets$ $f_{\rm new}$ + ((1/n.innerEdges) $*$ \hspace*{42pt} n.outerEdges);
		\EndIf
		\EndFor
		\For {$n \in {\rm neighbor\_ring}$}
		\State\Call{SpreadUpdate}{$n$,$i_e$,[$n_{\rm new}$,$e_{\rm new}$,$f_{\rm new}$]};
		\EndFor
		\State $i_e \gets i_e + 1$;
		\EndFunction
	\end{algorithmic}
	\caption{\textit{Split update event} messages are created by the \textit{SplitEvent} method of the update state.\\
		\newline
		\footnotesize The prevOrNextFI method represents that outer edges are only counted for positive event node
		neighbors when the previous or next node in the neighbor-ring also has a FI-value of one. 
		The neighbor-ring variable here is the list of neighbor sensor-IDs, and the pos- and neg-methods 
		represent the determination of a neighbor's node type as positive and negative event node respectively. 
		A sensor's FI-value and its event time stamp are used for the determination of the node type.}
	\label{ch4:alg:split_update_event}
\end{algorithm} 

\begin{lemma}\label{ch4:lemma:split_update}
Split updates correctly update regions after split events, independent on
whether other events occur concurrently inside the split regions.
\end{lemma}

\begin{proof}
First, the split components' numbers are reset to zero by spreading split update messages in the 
affected regions. All neighboring regions which were concurrently merged into one of the split regions 
will also receive split update messages, and will subtract the same amounts from their component 
numbers. Effectively, the component numbers of the split region in the previous state $S_{t-1}$
are subtracted from all component nodes:
To receive a split update message, an update node must either be part of the
targeted component, in that case its numbers are set to zero, or
it receives the split message from a different region through a merge event 
node, in which case the split message's numbers will be canceled by a merge message 
of the same amount. In either case, the split components' numbers will be erased from
all network nodes.
\par
During the second split phase each node sends split update event messages. 
Each update node counts one vertex, two connected nodes count a half of an edge each, and three
connected nodes each count a third of one face.
The overall sums of split update event messages will represent the split components’
numbers after split events. 
Due to the target-ID only regions surrounding split event nodes during the state $S_{t-1}$ 
are affected. Regions merged to concurrent split regions will not be recomputed.
\end{proof}

\begin{figure}[t]
	\centering
	\includegraphics[width=80mm]{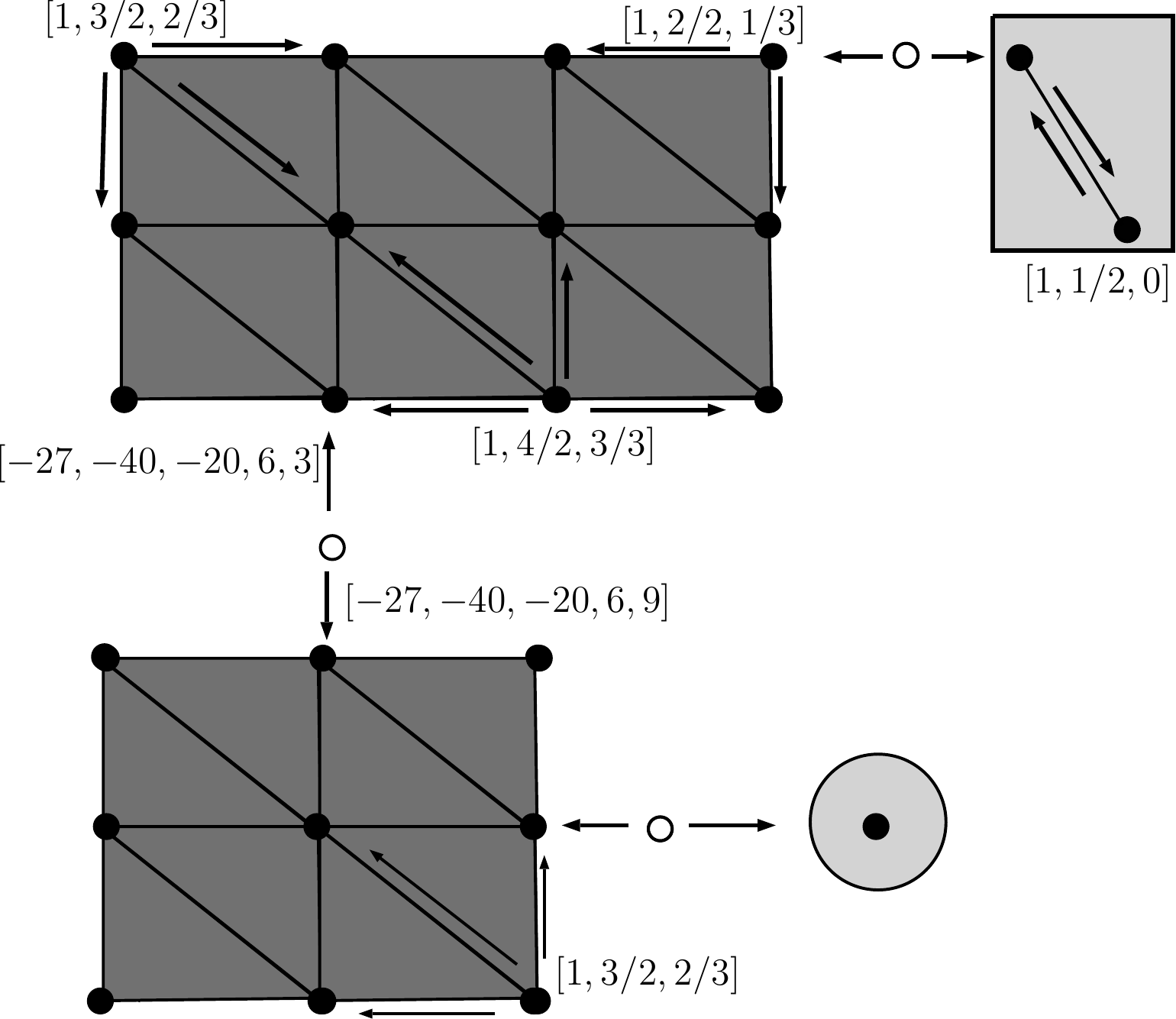}
	\caption{\footnotesize Three concurrent split events. The components' nodes
		reset their numbers via received split update messages
		before 	sending own \textit{split update event messages} to recompute 
		their components.}\label{f-sue}
\end{figure}

\begin{figure}[b]
	\centering
	\includegraphics[width=45mm]{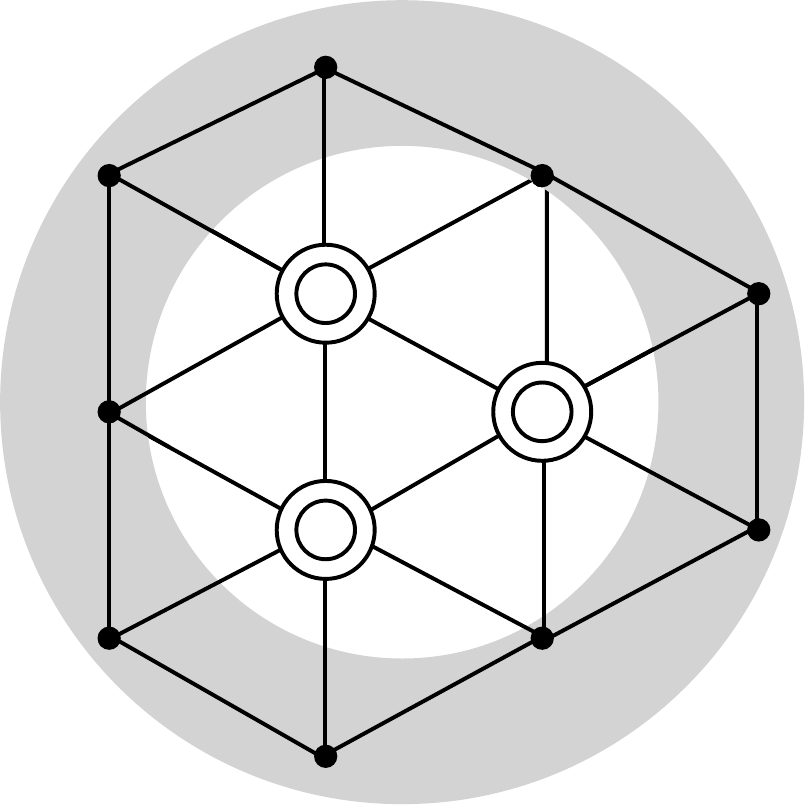}
	\caption{\footnotesize A hole appearance event with an event region 
		of three nodes. The surrounding region ring has 
		still the values $(n=12,e=24,f=13)$, and all three event nodes determine $({\rm In}(e)=4, {\rm Out}(e)=3)$ and a topological invariant event
		with: $\Delta\beta_0 = 0, \Delta\beta_1 = (-11+20-10+1) - (-12+24-13+1) = 0$.}\label{ch4:fig:inc}	
\end{figure}

\subsubsection{Event Regions} \label{ch4:event_region}
Although the concurrent detection of multiple events inside a connected component
of $\Delta_\theta$ is possible, regions consisting of event nodes cannot reliably
determine event types. An \textit{event region} is defined as a connected component of at least two 
network nodes where each contained node, independent on its FI-value, is an event node. 
Figure~\ref{ch4:fig:inc} demonstrates one such region of three negative event nodes.
In this particular example none of the three event nodes is able to detect a hole appearance.
Event nodes inside an event region cannot correctly determine their inner/outer edge numbers 
(Sect.~\ref{ch3:event_hom}) as the neighbor-rings' values do not reflect the previous network state. 
To enable event detection for event regions, event regions have to be replaced
with single event nodes. For this, each event node maintains a counter which indicates the sampling
interval at which the last event was detected. This time stamp can be passed along in ring
queries (Sect.~\ref{ch3:ring_query}). Neighboring event nodes send back blocking messages if both nodes have
the same time stamp and the querying node has a higher sensor-ID. Event nodes stop their event detection process upon
reception of such blocking messages and will retry the event detection in the next sample interval, 
provided their FI-values stay unchanged.
Event regions are thus replaced by single event nodes and their events' 
are reconstructed as chains of multiple sub-events.

\begin{theorem}\label{ch4:thm:r_update}
The region update correctly updates all component data for each node
inside $\Delta_\theta$ after non-incremental events detected at time $t$ 
when the sample interval is long enough to process all update messages.
\end{theorem}
\begin{proof}
All update messages sent during the region update process contain component
number differences. Therefore the messages' order of arrival/processing has no impact
on the resulting component numbers after the region update. 
By Lemma \ref{ch4:lemma:normal_update} \textit{normal updates} reach all nodes inside the affected regions
and are applied once at each node. Each normal update applies the respective event node's numbers
to its surrounding component.
By Lemma \ref{ch4:lemma:merge_update} \textit{merge updates} increment all merged components' numbers to the
component numbers of the merged regions. Due to the target-ID no
component is counted twice; each components' numbers are added once at
each node inside merged regions.
\par
By Lemma \ref{ch4:lemma:split_update} \textit{split events} cause a recomputation of all split regions.
Overall, the list $-[n_c,e_c,f_c]$ of previous component numbers in 
state $S_{t-1}$, and the lists of split components' new numbers $[n_d,m_d,f_d]$
are distributed into each component $d$. Regions concurrently
merged into split regions additionally will receive merge updates
$[n_c,m_c,f_c]$ based on the previous network state. The split update's
numbers will cancel the previous component's numbers at each node.
Only the component numbers $[n_d,m_d,f_d]$ remain.
\par
The usage of unique sensor-IDs for component-IDs guarantees that no two components
can receive the same component-IDs during one sample interval. However, it is possible
that two different components are created with the same sensor-IDs during
different network states. When the event node creating the event-ID is the same
in both states, the event number $i_e$ of the event-ID in $c_e$ will be different.
When the event nodes are different, the event node sensor-IDs will be different.
In either case the component-ID will be unique for all components.
\par
One cause for false component information would be neighboring event nodes, in which 
case the differences to previous component numbers could not be inferred reliably at event nodes. 
However, an event node cannot have a neighboring event node, since by (Sect.~\ref{ch4:event_region}) event regions are replaced by single event nodes. 
When all update messages of a region update are processed, the network's state is changed
from $S_{t-1}$ to $S_t$.
\end{proof}

\begin{figure}[t]
	\centering
	\includegraphics[width=70mm]{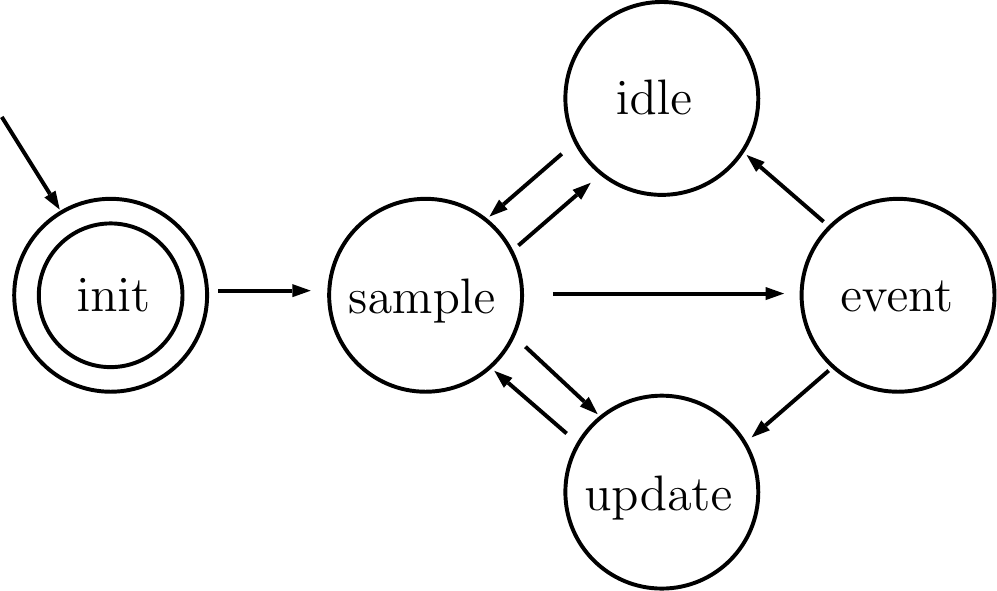}
	\caption{The FSM for topological event detection.}\label{ch4:fig:fsm}
\end{figure}
\subsection{Sensor FSM Model}\label{ch4:fsm}

A sensor designated for topological event detection can be described as a finite state machine (FSM) 
with five states (Fig.~\ref{ch4:fig:fsm}):
Initially no network data is available and all sensors are in the \texttt{init} state. 
Here the cyclic ordering of neighbor nodes necessary for event detection is established. 
Next follows the \texttt{sample} state. 
FI-values are sampled and the state changes accordingly. 
In case of a detected change of the FI-value, the \texttt{event} state is entered.
After event processing, an event node either is set inactive, 
i.e., it changes into the \texttt{idle} state, or the \texttt{update} state is entered.
In both cases the sample state is reentered at the beginning of the next sample interval.

\subsubsection{Init}
The network's communication structure is set up in the \texttt{init} state.
For this, sensors exchange their neighbor lists to determine lists of \textit{shared neighbors}. 
Then each sensor creates its own cyclic ordering of direct neighbors by passing a cycle message around its neighbor-ring; 
i.e., neighbors pass the message to the next neighbor shared with the initiating sensor until a cycle is completed.

\subsubsection{Sample}
At the beginning of each sample interval all sensors reside in the \texttt{sample} state. 
All sensors sample new FI-values. When a changed FI-value is detected, a sensor
changes into the \texttt{event} state. Otherwise, depending on the sensor's current
FI-value, either the \texttt{idle} state  ($\hat t_s=0$), or the \texttt{update}
state  ($\hat t_s=1$) is entered.

\subsubsection{Idle}
Upon entering the \texttt{idle} state, a sensor becomes inactive.
Only queries from neighboring event nodes are answered during the current sample interval. 
Event nodes in this context also include nodes in the update state which recompute their components
after split events (Sect.~\ref{ch4:split_update}).
The idle state is automatically left when a new sampling interval starts.

\subsubsection{Event}
The event state can be subdivided into three phases. 
First, all event node neighbors are queried in the form of a \textit{ring query} to update the neighbor-ring. 
At this point an event node's detected event can be canceled out by a neighboring event node 
with a lower sensor-ID (Sect.~\ref{ch4:event_region}).
Whereupon the sensor changes into the \texttt{idle} or \texttt{update} state, depending on its
previous FI-value, its newly sampled FI-value is ignored and all data
of the previous network state are retained. After a successful ring query and the determination of 
the event type, update messages are sent to the neighboring sensors in order to update the event
node's surrounding components (Sect.~\ref{ch4:r_update}).

\subsubsection{Update}
All sensors in the update state have a current FI-value of one; 
that is, they belong to a connected component of a monitored fire region.
A sensor in the {\tt update} state processes region update messages (Sect.~\ref{ch4:r_update}) sent by event nodes. 
Upon reaching a new sample interval, the update state is left.

\subsection{Self-Split Event Detection}\label{ch4:self_split}
Although self-split events cannot be distinguished from split events (Sect.~\ref{ch3:ring}), 
they can be detected as a byproduct of the split update process. 
When a self-split event happens, the detecting event node will pass a split update message to each ring component. 
Line~\ref{ch4:alg:split_update:ID} of Algorithm~\ref{ch4:alg:split_update} reveals that, instead of using the event node's sensor-ID, 
a representative node's sensor-ID is used for the region's component-ID after the split. 
For a self-split event this means that multiple component-IDs will be spread in the same region. 
\par
Split messages which reach already recomputed nodes (i.e., the event ID is known) will still be processed if they have a lower component-ID. 
It is exactly here, where a node changes its component-ID again following the same split event, that self-split events can be detected.
The original event node will receive a special self-split message from a neighboring node detecting the event during its update state (Fig.~\ref{ch4:fig:self_split_update}).
\par
The above described method for the detection of self-split events can fail to detect self-splits in cases where other split/merge events
happen concurrently in the same region. 
When another split/merge event node in the same region transmits the lowest component-ID, that ID could be spread to all ring components before the self-split can be detected. 
In that case the self-split event is detected as a split event at the event node. 
However, even then the self-split event can be detected indirectly after the region update process.
The component numbers after the split update will reveal a decrease in the first Betti number for self-split events.
\begin{figure}[!b]
	\centering
	\includegraphics[width=60mm]{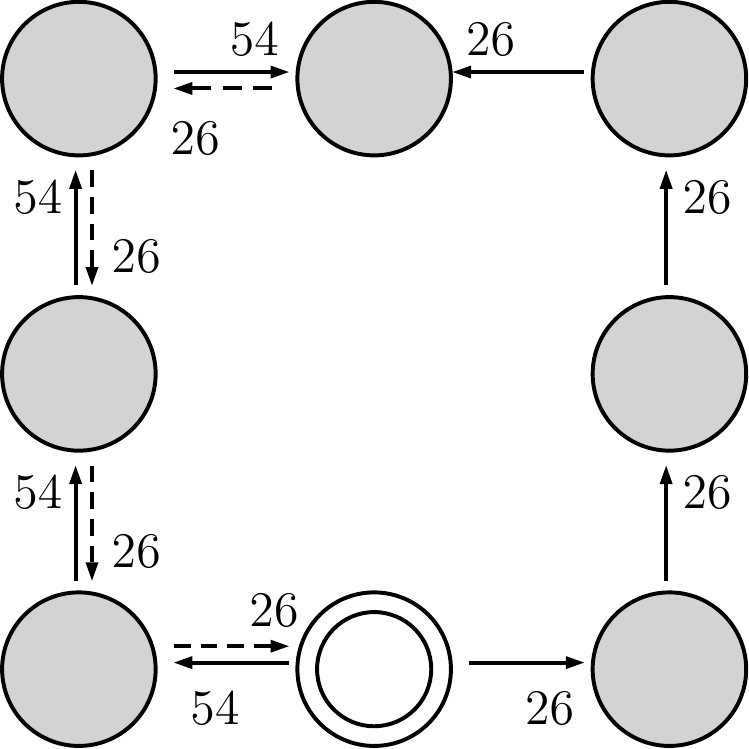}
	\caption{\footnotesize An example of a self-split event.
		The event node detects a split event,
		and transmits two split updates into both ring components.
		At the upper node the two new component-IDs 26 and 54 are received, but 
		only the lower value is used. The split message with ID 26 is passed
		along the other side of the region ring, changing the nodes' component 
		IDs from 54 to 26. At last, node 54 informs the event node of the detected self-split event.}
	\label{ch4:fig:self_split_update}
\end{figure}

\begin{figure}[t]
	\centering
	\includegraphics{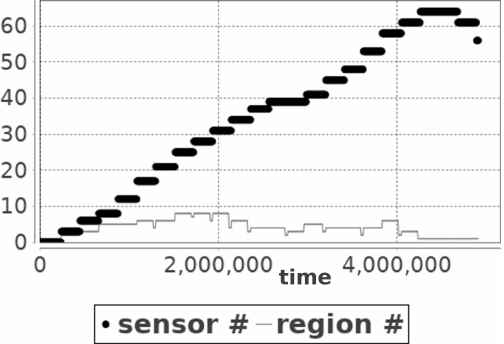}
	\includegraphics{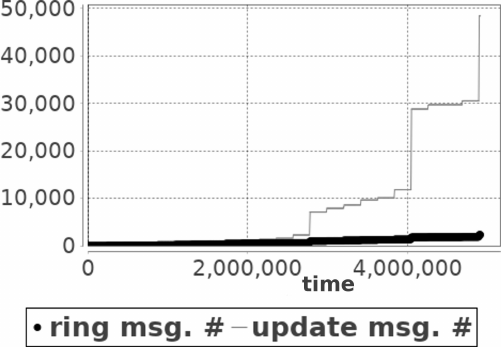}
	\caption{Charts of a network's message/region numbers during a TopED simulation of a spreading fire.}\label{ch5:fig:sim}
\end{figure}

\section{Simulations}\label{ch5}

The \textsf{GAMA}~\cite{gama} based multi-agent simulation tool \textit{TopED} (\textit{Top}ological \textit{E}vent \textit{D}etection) was 
developed in conjunction with this paper for the simulation of topological event monitoring. 
In the simulation network sensors are deployed in hexagonal grids, each sensor is located at the center of a hexagon, and communication 
links are established between sensors of neighboring hexagons. 
The resulting triangulation is a Whitney triangulation fulfilling all criteria defined in Section~\ref{ch3:sensor_model}.
In addition to the sensor network, randomly spreading forest fires are simulated as the to be observed regions. 
Coordinates of the simulated world contain FI-values  reflecting temperatures as scalars. 
Sensors sample the FI-values and convert them into binary values. 
The resulting subgraph of all sensors with FI-values of~one represents the simulated forest fire. 
\par
Supported by the TopED simulation are the nine topological event types discussed in Section~\ref{ch3:event_hom}. 
Figures~\ref{ch5:fig:fire1} and~\ref{ch5:fig:fire2} show examples of networks simulated with TopED. 
The underlying scalar field is represented as a heat map. 
Black regions represent fires, bright nodes indicate sensors with FI-values of~one, and dark sensors have FI readings of~zero. 
Additionally, communication links together with component-IDs are displayed for sensors lying inside fire regions.
Topological events are highlighted with event numbers, subscripts indicate the Betti number differences $(\Delta\beta_0, \Delta\beta_1)$,
and positive/negative events are indicated by bright/dark numbers.
Most computations during event monitoring (event/update state) are simple additions, and event nodes transmit messages only to a few select representative nodes.
Thus the complexity of computation at single network nodes is negligible. 
Important as a metric of complexity for distributed computations is the overall amount of transmitted messages:
\par
All in all approximately $(e \cdot n_r)$ messages are transmitted during \textit{ring queries} (Sect.~\ref{ch3:ring_query}): 
Each node in a query chain transmits one message, two chain ends are sent back per ring component and two messages are sent to zero nodes 
or detect a cycle per ring component. 
That makes for a total of $c_r \cdot (r_n + 4)$ messages, where $c_r$ denotes the number of ring components and $r_n$ represents the average length of a ring component. 
Not counting the number of zero nodes directly queried by an event node, this number more or less is equal to the number of event node neighbors $n_r = |N_s|$. 
Hence, with $e$ events happening during one sample period, ($e \cdot n_r$) ring messages are transmitted.
\par
During each sample period of a network approximately $(e \cdot n_r \cdot n_c)$ \textit{region update messages} (Sect.~\ref{ch4:r_update}) are transmitted in total: 
Each update node passes one update message to each of its neighbor nodes. 
Let $n_c$ be the number of non-event nodes lying in positive regions in which at least one event happened, 
$e$ be the number of events, and $n_r$ be the average number of neighbors of each node. 
Then $(e \cdot n_r \cdot n_c)$ update messages are transmitted during one sample period.
\par
When split events occur, update nodes additionally execute ring queries to recompute their components' numbers, 
and create \textit{split update event messages} (Sect.~\ref{ch4:split_update}). 
All affected update nodes start additional region updates with message amounts comparable to regular detected events. 
In such cases the update message number can increase up to $(e \cdot n_r \cdot n_c) + (e \cdot n_r \cdot n_c^2)$,  
while the number of ring messages can increase up to $(e+(e \cdot n_c)) \cdot n_r$.
\par
The charts in Figure~\ref{ch5:fig:sim} display the amount of transmitted messages during the simulation of a spreading forest fire with TopED.
In the lower chart the number of ring and update messages are shown separately, 
while the upper chart displays the number of sensors lying inside fire regions in relation to the number of different detected regions. 
The amount of update messages clearly exceeds the number of ring messages;
update messages constitute the largest part of in-network communication.
\par
Figure~\ref{ch5:fig:fire2} shows the first detected case of merging regions during the simulation, 
it is at this point where a first visible small increase in update messages becomes visible in Figure~\ref{ch5:fig:sim}.
Afterwards, the message number steadily increases with growing region sizes.
Evidently, the variable $n_c$, the average region size, is the deciding factor for the message amount. 
And Figure~\ref{ch5:fig:fire1} highlights the worst case scenario: split updates. All three visible steps in the update message graph are the result of split events occurring.
\begin{figure}[t]
	\centering
	\includegraphics{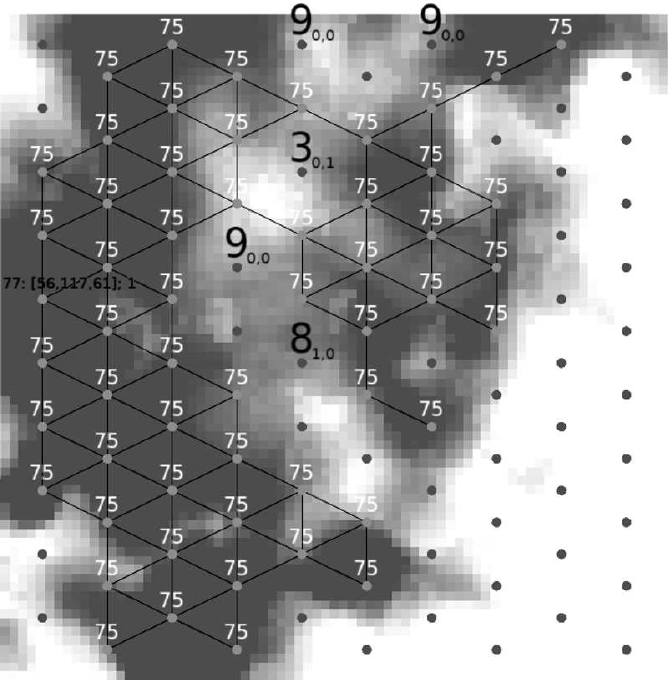}
	\caption{
		A hole-appearance, a self-split and three topologically invariant events are detected.}\label{ch5:fig:fire1}
\end{figure}
\begin{figure}[t]
	\centering
	\includegraphics{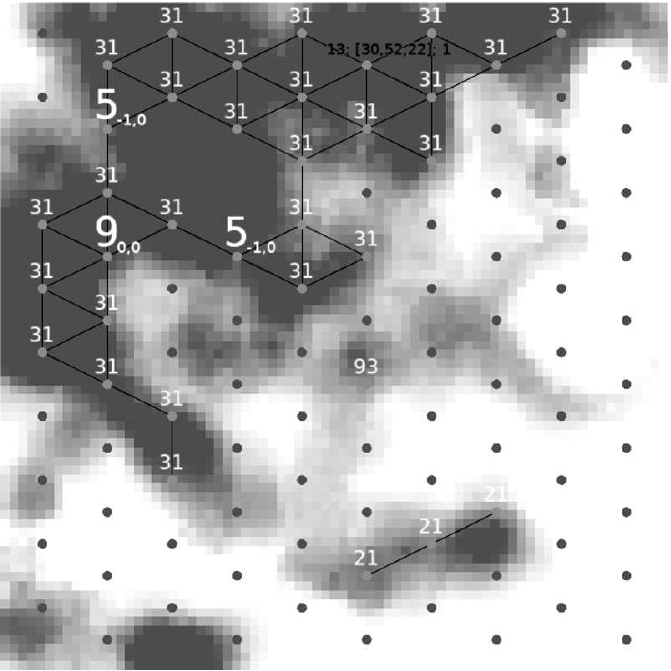}
	\caption{
		Two merge events of the same region occur concurrently, plus one topologically invariant event.}\label{ch5:fig:fire2}
\end{figure}

\section{Conclusion}\label{ch6}

After the definition of Betti numbers and their computation for graphs associated to Whitney
triangulations in Section~\ref{ch2}, 
the topological event detection model central to this paper was defined in Section~\ref{ch3}. 
Events are detected locally via homology computations at event nodes. 
In total nine basic event types are supported. Component numbers necessary for 
graph homology computation are updated via region updates. 
The usage of differences for update messages (Sect.~\ref{ch4:r_update}) and the reduction of
event regions (Sect.~\ref{ch4:event_region}) to single event nodes allows for the 
conflict-free detection of non-incremental events.
\par
We assumed a simplified communication model with globally synchronized clocks and error-free 
communication for event detection. This model was used to verify the correct monitoring 
of topological events in simulation. In practice locally synchronized clocks are used for 
distributed computations. For application in real-world examples (e.g., in WSNs) 
the definition of an extended message model dealing with communication errors will be required. 
\par
The last section showed that the complexity of topological event monitoring lies in the region update
process, split updates in particular cause a great message amount. It was assumed that all update nodes
execute ring queries when performing a \textit{split update} (Sect.~\ref{ch4:split_update}). 
Most of the nodes' neighbor values are unchanged when compared to the last sample interval of the network. Additional ring queries are unnecessary for these nodes. 
An extension to the topological event detection model could include event nodes notifying 
their neighbors of their changed FI-values during ring queries. 
With this additional bit of information update nodes not only could avoid the execution of 
additional ring queries during split updates, but also could transmit update messages exclusively
to neighbor nodes with FI-values of~one.
\par
Section~\ref{ch3:event_tree} introduced the event decision diagram and stated that Betti number
differences can be inferred from the neighbor-ring.
When neither the components' network numbers $n_c,m_c,f_c$ nor their exact Betti numbers
are required, the neighbor-ring alone is sufficient for event detection. 
Ignoring the actual Betti numbers, region updates are then reduced to the spreading of 
component-IDs after merge/split events, 
making the component recomputation during the \textit{split update} (Sect.~\ref{ch4:split_update}) superfluous.

\bibliographystyle{acm}
\bibliography{bibliography}

\end{document}